\documentclass[12pt]{article}

\usepackage{fullpage, amsmath, amssymb, amsthm, bbm, color, dsfont}
\usepackage{graphicx,caption,subcaption,placeins}
\usepackage{setspace}
\usepackage{comment}

\usepackage[colorlinks = true, citecolor = blue,
linkcolor = blue]{hyperref}
\usepackage{cleveref}
\usepackage{parskip}

\makeatletter
\def\verbatim@font{\linespread{1}\normalfont\ttfamily}
\makeatother

\usepackage{natbib}
\bibpunct{(}{)}{;}{a}{,}{,}
\usepackage{enumerate}
\newtheorem{assumption}{Assumption}
\newtheorem{definition}{Definition}
\newtheorem{theorem}{Theorem}
\newtheorem{proposition}{Proposition}
\newtheorem{theorem*}{Theorem}
\newtheorem{proposition*}{Proposition}
\newtheorem{corollary*}{Corollary}
\newtheorem{procedure}{Procedure}
\newtheorem{example}{Example}

\newtheorem{lemma}{Lemma}
\newtheorem{corollary}{Corollary}
\newtheorem{remark}{Remark}

\newcommand{\nn}{\newline}
\newcommand{\Adom}{\mathbb{A}}
\newcommand{\Ldom}{\mathbb{L}}
\newcommand{\Wdom}{\mathbb{W}}

\newcommand{\iv}{\mathds{1}}

\newcommand{\Ib}{ \mathbb{U} }

\newcommand{\Aset}{\Adom}

\newcommand{\Hset}{\Wdom}

\newcommand{\Lset}{\Ldom}

\newcommand{\h}{c}

\newcommand{\nak}{\mathbf{n}}

\newcommand{\wsymb}{\omega}
\newcommand{\uw}{u^\wsymb}

\newcommand{\Yw}{Y^\wsymb}
\newcommand{\Tw}{T^\wsymb}
\newcommand{\fw}{f^\wsymb}

\newcommand{\fz}{f}

\newcommand{\lsymb}{\ell}
\newcommand{\ul}{u^\lsymb}
\newcommand{\wl}{w^\lsymb}

\newcommand{\fl}{f^\lsymb}

\newcommand{\ones}{\mathbf{1}}

\newcommand{\SR}{\mathrm{SR}}
\newcommand{\CR}{\mathrm{CR}}

 \renewcommand{\h}{\mathsf{w}}

\newcommand{\Yobs}{Y^{\textup{obs}}}
\newcommand{\Zobs}{Z^{\textup{obs}}}
\newcommand{\Wobs}{W^{\textup{obs}}}
\newcommand{\Tobs}{T^{\textup{obs}}}
\newcommand{\obs}{\textup{obs}}
\newcommand{\Uobs}{U^{\obs}}

\newcommand{\pr}{P}
\newcommand{\tS}{\textup{S}}
\newcommand{\tSm}{\textup{Sm}}

\newcommand{\tSL}{\textup{SL}}
\newcommand{\tSLm}{\textup{SLm}}

\usepackage{xr}

\newcommand{\TheoremOne}{
 Let $\pr (L)$ denote a distribution of the group labels with support
  $\Ldom = \{1, \ldots, K\}^N$. 
Let $W = \wl(L)\in\Wdom^N$ be the corresponding exposures, 
and let $U = \ul(L) \in\{0,1\}^N$ be the focal indicator vector, for some $\wl(\cdot), \ul(\cdot)$ defined by 
the analyst.
  Define $\mathbb{S}_{A,U} = \SN(A)\cap \SN(U)$, which is the permutation subgroup of $\SN$ that leaves $A$ (the attribute vector) and $U$ (the focal unit vector) unchanged. Suppose that the following conditions hold.
  \begin{enumerate}[(a)]
  \item \label{thm-a} %
     $\pr(L) = \pr(\pi L)$, for all $\pi\in\mathbb{S}_{A,U}$ and $L\in\Ldom$.

  \item\label{thm-b}
 $\wl(\cdot )$ is equivariant with respect to $\mathbb{S}_{A,U}$.  

    \item\label{thm-c}
$\ul(\cdot )$ is equivariant with respect to $\mathbb{S}_{A,U}$.  

  \end{enumerate}
  Then, $W$ is uniformly distributed conditional on the 
  event $\{W \in \mathcal{B}\}$, where $\mathcal{B}\in\mathcal{O}(\Wdom^N; \SAU)$.
}

\newtheorem{manualtheoreminner}{Procedure}
\newenvironment{manualproc}[1]{%
  \renewcommand\themanualtheoreminner{#1}%
  \manualtheoreminner
}{\endmanualtheoreminner}

\usepackage{xcolor}

\title{\bfseries Randomization tests for peer effects in group formation experiments}
\author{Guillaume Basse, Peng Ding, Avi Feller, Panos Toulis\thanks{
We thank Alex Franks, Xinran Li, Sam Pimentel, and Fredrik S{\"a}vje as well as seminar participants at UCLA, UC Berkeley, and UCSB for helpful comments. GB acknowledges support from the U.S. National Science Foundation (grant \# 1713152). PD acknowledges support from the U.S. National Science Foundation (grants \# 1713152 and \# 1945136). AF gratefully acknowledges support from a National Academy of Education/Spencer Foundation postdoctoral fellowship. PT is grateful for the John E. Jeuck Fellowship at Booth.}}

\begin{document}

\maketitle

\pagenumbering{gobble}

\begin{abstract}
Measuring the effect of peers on individuals' outcomes is a challenging problem, in part because individuals often select peers who are similar in both observable and unobservable ways. 
Group formation experiments avoid this problem by randomly assigning individuals to groups and observing their responses; for example, do first-year students have better grades when they are randomly assigned roommates who have stronger academic backgrounds?
In this paper, we propose randomization-based permutation tests for group formation experiments, extending classical Fisher Randomization Tests to this setting. 
The proposed tests are justified by the randomization itself, require relatively few assumptions, and are exact in finite-samples. This approach can also complement existing strategies, such as linear-in-means models, by using a regression coefficient as the test statistic.
We apply the proposed tests to two recent group formation experiments.
\medskip 

\end{abstract}


\onehalfspacing
\clearpage
\pagenumbering{arabic}


\section{Introduction}

Peers influence a broad range of individual outcomes, from health to education to co-authoring papers.\footnote{All of the co-authors entered the same graduate program in the same year.}
However, studying these peer effects in practice is challenging in part because individuals typically select peers who are similar in both observed and unobserved ways \citep{sacerdote2014experimental}.
\emph{Randomized group formation}, also known as exogenous link formation, avoids this problem by randomly assigning individuals to groups and observing their responses.
Among its many applications, this approach has been used to
assess the effect of dorm-room composition on student grade point average
\citep[GPA;][]{sacerdote2001peer,bhattacharya2009inferring,li2018randomization}, the effect of squadron composition on  individual performance at military academies 
\citep{lyle2009effects,carrell2013natural}, the effect of business groups on the diffusion of 
management practices \citep{fafchamps2015networks,cai2017interfirm}, 
the effect of group or team assignments on the performance of professional athletes \citep{guryan2009peer}, and the effect of co-workers on productivity \citep{herbst2015peer, cornelissen2017peer}. 
A typical substantive question is then: what is the effect of randomly assigning an incoming first-year student to a roommate with high academic preparation (the ``exposure'') on the student's own end-of-year GPA?

In this paper, we propose analyzing randomized group formation designs from the perspective of ``randomization inference,'' in the spirit of \citet{fisher1935}. 
Like the classic Fisher Randomization Test (FRT), our ultimate proposal is a straightforward permutation test that (conditionally) permutes each individual's exposure.
This test is exact in finite-samples, requires relatively few assumptions, and is justified by the randomization itself.
Thus, we argue that our approach is a natural benchmark for analyzing randomized group formation designs, building on a growing literature within economics and econometrics \citep[see][]{lehmann2006testing,imbens2015causal,canay2017randomization, young2019channeling} that seeks 
to use the randomization itself as the source of uncertainty when analyzing randomized trials.
Moreover, we can combine this approach with popular model-based frameworks, such as the linear-in-means model \citep{manski1993identification}, by using a model to generate the test statistics for subsequent randomization tests.  When such models are correctly specified, the corresponding randomization tests are likely to have higher power. Even when the models are incorrectly specified, our proposed randomization tests can still ensure that the $p$-values are finite-sample valid.

\pagebreak

To develop this procedure, we overcome several technical and computational hurdles.
First, a key challenge for randomization tests under interference is that the null hypotheses of interest are not typically ``sharp,'' in the sense of specifying all potential outcomes for all units \citep{rosenbaum2007interference, hudgens2008toward}. For example, the null hypothesis of no difference between having 0 or 1 students with high academic preparation in a dorm room does not have any information about dorm rooms that have 2 students of that type.
An important innovation for causal inference under interference is to restrict the randomization test to a subset of units, known as \emph{focal units}, which ``makes the null hypothesis sharp'' and allows for otherwise standard conditional randomization tests \citep{aronow2012general, athey2018exact, basse2019rand}. Our first contribution is to extend these results to randomized group formation designs,  and show that restricting our attention to focal units indeed enables valid randomization-based tests, at least in principle.

In practice, however, it is difficult to obtain draws from the appropriate null distribution in group formation designs. The computationally straightforward approach of naively permuting the exposure of interest (e.g., permuting the number of students in a room of a specific type) is not typically valid, since permuted exposures can be incompatible with the original group formation design.
Conversely, the conceptually valid approach of repeatedly assigning groups can be computationally prohibitive for testing non-sharp null hypotheses that require conditioning on a specific set of focal units.

Our second main contribution is therefore to develop computationally efficient randomization tests that can be implemented easily via permutations.
In particular, for a broad class of designs, we show that permuting exposures separately for each level of individuals' own attributes (e.g., high academic preparation) leads to valid randomization tests.
Using algebraic group theory, we prove that a key property in all these designs is {\em equivariance}, which, roughly speaking, ensures that an invariance in the design translates into an invariance on peer exposure.
%
 Our paper thus provides one of the first, general theoretical results on 
efficient implementation of randomization tests of peer effects via permutations.

We apply our results to two studies based on randomized group formation designs: freshmen randomly assigned to dorms \citep{li2018randomization} and chief executive officers (CEOs) randomly assigned to group meetings \citep{cai2017interfirm}. 
We describe stylized versions of these examples in the next section and discuss the applications in more detail in Section~\ref{section:applications}.
In the appendix, we also include extensive simulation studies showing both the validity of the method and its power under a range of scenarios. 

Our approach combines two recent strands in the literature on causal inference under interference. 
In the first thread, \citet{aronow2012general}, \citet{athey2018exact}, and \citet{basse2019rand} develop conditional randomization tests that are valid under interference;
we discuss this further in Section \ref{section:rand-non-sharp}. In that setup, the groups are fixed and the intervention itself is randomized.  
In the second thread, \citet{li2018randomization} explicitly consider group formation designs and define peer effects using the potential outcomes framework.  Their paper mainly  considers the {\it Neymanian} perspective that focuses on randomization-based point and interval estimation based on normal approximations \citep{imbens2015causal, abadie2020sampling}. By contrast, our paper chiefly considers the {\it Fisherian} perspective that instead focuses on finite-sample exact $p$-values via randomization-based testing. 
This allows us to examine hypotheses for smaller subpopulations, including those in our motivating examples. 
Moreover, our approach is valid for arbitrary outcome distributions, including possibly heavy-tailed sales revenue in the second example~\citep{rosenbaum2002observational, lehmann2006testing}.


\section{Setup and framework}
\label{section:setup}

\subsection{From regression to randomization inference for peer effects}
\label{section:regression}

To illustrate the notation and the key concepts, we introduce two running examples. 
Example~\ref{ex:sacerdote} presents an idealized version of \citet{sacerdote2001peer} and \citet{li2018randomization}, in which incoming college freshmen are randomly assigned to dorm rooms.
Example~\ref{ex:cai} presents an idealized version of \citet{cai2017interfirm}, in which 
CEOs of Chinese firms are randomly assigned to attend monthly group meetings. 
Both examples have a common structure in which individuals are randomly assigned to groups. We observe attribute $A$ and outcome $Y$ for each individual, and the attributes of peer individuals in the group, $W$. The goal is to estimate the ``effect'' of $W$ on $Y$.
We make these statements more precise in the next section and analyze the original data from both examples in Section~\ref{section:applications}.

\begin{example}
\label{ex:sacerdote}
Suppose that $N$ incoming freshmen are paired into $N/2$ dorm rooms of size $2$. We classify 
	incoming freshmen as having high ($A=1$) or low ($A=0$) incoming level of academic preparation (e.g., based on standardized test scores and high school grades). 
	We want to understand whether a freshman's end-of-year GPA varies based on the academic preparation of his or her roommate ($W$). 
	Specifically, is there an effect on end-of-year GPA ($Y$) of being assigned a roommate with `high' incoming preparation ($W = 1$) relative to being assigned to a roommate with `low' incoming preparation ($W = 0$)? 
\end{example}
\begin{example}\label{ex:cai}
	Suppose that $N$ firm CEOs are 
	assigned to $N/3$ monthly meeting groups of size $3$ where they discuss business and management practices. 
	Each CEO is classified as leading a `large firm' ($A=1$) or `small firm' ($A=0$). 
	We want to assess whether the revenue of a CEO's company ($Y$) is affected by the composition of the meeting group ($W$). 
	Specifically, is there an impact on the firm's revenue of assigning that firm's CEO to a group with two CEOs from large firms ($W = 2$) relative to assigning that firm's CEO to a group with one ($W = 1$) or no CEOs ($W = 0$) from large firms? 
\end{example}

These examples capture the notion of a peer effect as the idea that a given unit's outcome may be affected by their peers' attributes. A vast literature in economics formalizes these ideas; see, among others, \citet{manski1993identification}, \citet{brock2001interactions}, \citet{sacerdote2011handbook}, \citet{goldsmith2013social}, and \citet{angrist2014perils}. We now briefly review common existing approaches and discuss recent work that motivates the use of linear regression from the randomization perspective \citep{li2018randomization}. 
Since our eventual goal is a fully randomization-based framework for analyzing randomized group formation designs, our discussion here necessarily focuses on reduced-form approaches, setting aside a vibrant literature on more structural models of peer effects and social interactions \citep[see][]{bramoulle2020peer}.

\paragraph{Linear-in-means model.}
We begin with the  workhorse \emph{linear-in-means model}, described in detail in a seminal paper from \citet{manski1993identification}, which regresses $Y$ on $\overline{A}$, the average attribute in the group. Following a long literature \citep[see][]{sacerdote2011handbook}, we initially consider the leave-one-out form of this model, which separates out $A$, a unit's own attribute, and $W$ (a transformation of) the leave-own-unit-out average attribute:
$$Y_i^{\text{obs}} = \alpha + \beta A_i + \tau W_i + \varepsilon_i,
$$
where $Y_i^{\text{obs}}$ is the observed outcome for unit $i$.
For Example 1, both $A$ and $W$ are binary; for Example 2, $A$ is binary and $W$ takes on three values, $\{0,1,2\}$. The coefficient $\tau$ is referred to as the \emph{exogenous peer effect} \citep{manski1993identification} or the \emph{social return} \citep{angrist2014perils}.
Standard errors are typically clustered at the group level.
Importantly, we do not include specifications with $Y$ on the right-hand side and therefore do not consider so-called \emph{endogenous peer effects}. While this avoids a range of thorny econometric questions \citep[see][]{manski1993identification, angrist2014perils}, this choice necessarily restricts the type of substantive questions we can address. 
Similarly, since we focus on experiments in which individuals are randomly assigned to groups, we also exclude \emph{correlated effects}, which could arise if individuals self-select into groups.

Interestingly, \citet{kolesar2011inference} and \citet{angrist2014perils} note the connection between this linear regression and a jackknife instrumental variables estimator (JIVE), with group as the instrument. Let $L$ be the group indicator, then the coefficient $\tau$ above is equivalent to the leave-own-unit-out two-stage least-squares coefficient of (informally) 
$Y$ on $A$, instrumented with $L$. In noting this connection, \citet{angrist2014perils} argues that the many weak instruments problem is partly responsible for the poor finite-sample behavior of regression estimators --- a behavior we also observe in our applications.

\paragraph{Heterogeneous treatment effect model.} Even focused exclusively on exogenous peer effects, there are many challenges with the linear-in-means model. Most immediately, as \citet{sacerdote2011handbook} notes: ``from an empirical point of view, researchers have found that peer effects are not in fact linear-in-means''. This has led researchers to instead consider interacted specifications that allow for possible nonlinearities \citep{sacerdote2001peer, duncan2005peer, cai2017interfirm}. In the context of our examples these are specifications of the form:
\begin{equation}
\label{eq::ols-interaction}
Y_i^{\text{obs}} = \alpha + \beta A_i + \tau W_i + \gamma A_i \cdot W_i + \varepsilon_i
\end{equation}
Here the relevant effects are appropriate combinations of the coefficients $\tau$ and $\gamma$, and, as above, the standard errors are typically clustered at the group level. 
Again, this interacted model is typically motivated by the desire to estimate a more flexible specification for the (sometimes implicit) underlying model of social interactions.

\paragraph{Motivating regression from randomization.} Somewhat surprisingly, \citet{li2018randomization} show that, for a broad class of randomized group formation designs, randomization fully justifies the interacted specification \eqref{eq::ols-interaction} above. Moreover, \citet{li2018randomization} argue that the randomization-based perspective justifies the use of \emph{non-clustered} robust standard errors, suggesting that the common practice of clustering standard errors is overly conservative for such designs, analogous to arguments from \citet{abadie2023cluster}.
In this case, failing to include the interaction (i.e., simply running the regression of $Y$ on $A$ and $W$) leads to a precision-weighted average of the subgroup effects, though this approach is no longer equivalent to a randomization-based estimator.

\paragraph{From regression to randomization-based testing.}
As we show below, the regression-based approach from \citet{li2018randomization}, while conceptually elegant, can have poor finite-sample performance.
In particular, the asymptotic theory in that paper assumes that both $A$ and $W$ have very few levels, and that the number of individuals within each $A \times W$ group is large. This is not a reasonable approximation in our applications, however; 
for instance, in the roommates application we analyze in Section \ref{section:applications}, the size of an $A \times W$ subgroup can be as small as four students.

Our main contribution is to justify and implement randomization-based tests for exogenous peer effects, building on recent proposals for randomization tests under interference \citep{aronow2012general, athey2018exact, basse2019rand, puelz2022graph}.
At a high level, we propose the permutation-based analog of the fully interacted regression model discussed above. The primary technical obstacle is justifying this approach from the randomized group formation design itself. As we will see, this requires substantial technical overhead, even if the final procedure is itself straightforward. To demonstrate this, we also develop theory for general randomization-based tests for non-sharp nulls.

\subsection{Notation and setup}
\label{section:preliminaries}
We now formalize the problem setup outlined above.
Consider $N$ units to be assigned to $K$ different groups; both numbers are fixed. Let $\Ib = \{1, \ldots, N\}$ denote the set of units.
Let $L_i \in \Ldom = \{1, \ldots, K\}$ denote the labeled group to which unit $i$ is assigned, and define $L = (L_i)_{i=1}^N$ as the full group-label assignment vector. Also, let $\pr(L)\in[0,1]$ denote the 
probability distribution of $L$, which is known from the experimental design. 
In a group formation design, the individual $i$'s treatment assignment can be defined as
\begin{align}\label{eq:ZL}
	Z_i & =  \big\{j \in \Ib  : j\neq i \text{ and } L_j = L_i \big\}.
\end{align}
Assignment $Z_i$ is therefore the set of individuals assigned to the same  group as individual $i$. Let $Z = (Z_i)_{i=1}^N$ be the full assignment vector.

As we discuss above, a key feature of our setting is that each individual $i$ exhibits a salient \emph{attribute}, $A_i$; for example, $A_i = 1$ if individual $i$ has high academic preparation entering college. 
This attribute often plays a special role in group formation designs; for example, in the \emph{stratified group formation design} we consider in Section \ref{sec:sr_designs}, a room must have a fixed, pre-defined number of students with $A_i = 1$.
Formally, attribute $A_i$ takes values in a set $\Adom$, which could be a transformation (e.g., coarsened version) of covariates $X_i$.~We let $A = (A_i)_{i=1}^N$ and $X= (X_i)_{i=1}^N$ be the full vector of attributes and matrix of covariates,  respectively. 

The goal of this paper is to understand how peers' attributes affect unit outcomes, and so we define the {\em exposure} for each unit $i$ as: 
\begin{equation}\label{eq:exposure}
	W_i = w_i(Z) = \{A_j: j \in Z_i\},
\end{equation}
that is, the exposure of unit $i$ is the multiset of attributes of its neighbors, where a multiset is a set with possibly repeated values. Define 
$W = w(Z) = (w_i(Z))_{i=1}^N$ as the full vector of exposures, and denote by 
$\Wdom = \{  \h_1, \ldots, \h_m  \}$ the finite set of possible exposure values in the experiment. 
Finally, we let $Y_i(Z)$ denote the real-valued potential outcome of unit $i$ under assignment $Z$.

While this formulation is general, it is often useful to define exposures as simple functions of the attribute vector $A$. 
For example, when $A$ is binary, a natural choice is to define
\begin{equation}\label{eq:exposure-2}
W_i = 	w_i(Z) = \sum_{j\in Z_i} A_j,
\end{equation}
the number of ``neighbors" of unit $i$ with attribute $A = 1$. 
All results in the paper hold for general exposure mappings as in \eqref{eq:exposure}; we use the simpler formulation in \eqref{eq:exposure-2} in the running examples for simplicity.

\paragraph{Notation.}~These definitions are nested, so that $L$ determines $Z$, and $Z$  determines $W$.
As such, any function on one domain is also a function on a `finer' domain.
To ease notation, we will use `$\fz(Z)$' to denote a  function defined on the domain of $Z$ that is implied by $\fw(W)$, and, similarly, use `$\fl(L)$' to denote 
the function on the domain of $L$ that is implied by either $\fw(W)$ or $f(Z)$, noting that these all map to the same value:
$
\fw(W) = \fz(Z) = \fl(L).
$
For instance,  we write $W = \wl(L)$ to express the exposures in~\eqref{eq:exposure} as a function of $L$.

\subsection{Assumptions and exclusion restrictions}
\label{sec:assumptions}
The primary goal of our analysis is to estimate the causal effect of exposing a unit to a mix of peers with one set of attributes versus another, known as the \emph{exogenous peer effect} \citep{manski1993identification} or the \emph{social return} \citep{angrist2014perils}.
Formalizing such effects is non-trivial, however, with a substantial literature defining estimands in terms of coefficients in a linear model. 
Following a more recent set of papers, we instead formalize these effects via exposure mappings based on potential outcomes \citep{toulis2013estimation, manski2013identification, aronow2017estimating, li2018randomization}, which capture the summary of $Z$ that is sufficient to define potential outcomes on the unit level.

To do so, we make the critical assumption that the exposure is \emph{properly specified} in the sense defined below~\citep{aronow2017estimating}:

\begin{assumption}\label{asst:properly-specified}
	For all $i\in \Ib$ and for all $Z,Z'$, we have 
	\begin{equation*}
 w_i(Z) = w_i(Z') \Rightarrow Y_i(Z) = Y_i(Z').
	\end{equation*}
\end{assumption}

Under Assumption~\ref{asst:properly-specified}, each unit $i$ has $|\Wdom| = m$ potential outcomes, one for 
each level of exposure, and we may write
$$
Y_i(Z) = \Yw_i(w_i(Z)) = \Yw_i(W_i)
$$
to indicate that potential outcomes depend only on the exposure level and not the particular group assignment.

\addtocounter{example}{-2}
\begin{example}[continued]
	With dorm rooms of size $2$, the exposure $W_i$ of student $i$ is then the attribute $A_j$ of student $i$'s roommate. More generally, under 
	the exposure mapping in \eqref{eq:exposure-2}, each unit has only two possible exposures, 
	since $W_i \in \Hset = \{0,1\}$, and thus each unit has two potential outcomes $\{  \Yw_i(0), \Yw_i(1) \} $. 
\end{example}

\begin{example}[continued]
	Here, each group has size $3$ and the assignment $Z_i$ of unit $i$ is the unordered pair of indices of the other two CEOs in the group. CEO $i$'s exposure is then the number of the other CEOs from large firms.
	In this 
	case, each unit has three possible exposures, since $W_i \in \Hset = \{0, 1, 2\}$ under \eqref{eq:exposure-2}, and thus each unit has 
	three potential outcomes $\{ \Yw_i(0), \Yw_i(1), \Yw_i(2) \} $. 
\end{example}

\paragraph{Discussion of Assumption \ref{asst:properly-specified}.}
Assumption \ref{asst:properly-specified}, which is \emph{not} justified by the randomization, is the key substantive assumption in our setup and merits further discussion. At its core, this assumption is an \emph{exclusion restriction}: the only impact of the randomization on an individual's outcome is by changing the salient attributes $A$ --- and only the salient attributes --- of the other individuals in the group. 
For instance in 
Example~\ref{ex:sacerdote}, Assumption \ref{asst:properly-specified} implies that room assignment affects unit $i$'s freshman GPA only by changing $i$'s roommate's academic ability, excluding other possible channels of peer influence. This necessarily reduces otherwise complex individual and social interactions to a scalar quantity; for discussion, see \citet{sacerdote2011handbook}.\footnote{Similar challenges arise in other econometric applications, such as `judge fixed effects', where the choice of attribute (e.g., conviction rate) is important in the overall analysis \citep[e.g.,][]{frandsen2023judging}.} 
Assumption \ref{asst:properly-specified} also plays a role analogous to the stable unit treatment value assumption (SUTVA) by ruling out effects from changing \emph{other} groups.
Thus, when combined with the exposure mapping of \eqref{eq:exposure}, 
this assumption implies both a form of \emph{partial interference} 
and a form of \emph{stratified interference} \citep{hudgens2008toward}. 
Finally, beyond assuming that attribute $A$ is the relevant scalar quantity, Assumption \ref{asst:properly-specified} also assumes that the functional form is correctly specified, though we typically allow $W$ to be fully flexible with respect to $A$.

As we discuss in Appendix \ref{section:assumptions-hypothesis}, the procedure we outline below will still lead to a valid test without imposing Assumption \ref{asst:properly-specified} --- though interpreting that rejection is challenging. In particular, the test might reject if the null hypothesis is indeed correct but Assumption \ref{asst:properly-specified} does not hold, for instance if an individual's outcome depends on attributes other than $A$. 
At present, there is little guidance for applied researchers on specifying exposure
mappings, in part because these mappings can be highly context-dependent. 
For point estimation, violating Assumption \ref{asst:properly-specified} complicates the implied estimand, which will typically correspond to a particular weighted average of treatment effects.
See \citet[Section 7]{li2018randomization} for a discussion in the context of peer effects; \citet{savje2021causal} for a more general discussion
of inference with misspecified exposure mappings; and \citet{leung2022causal} for an alternative approach that considers approximate exposures.
The situation is more complicated for testing, where it is difficult to interpret a rejection in the absence of Assumption \ref{asst:properly-specified}. This remains an open research area.  As one possible direction forward, see recent work from \cite{hoshino2023randomization}, who propose randomization-based specification tests for exposures.


\subsection{Sharp and non-sharp null hypotheses}
\label{section:null-hypotheses}

Following the literature on FRTs, we focus on hypotheses defined at the unit level, unlike the regression-based approaches in Section \ref{section:regression}, which focus on so-called \emph{weak null} hypotheses that average over units.
A key technical challenge is that many unit-level null hypotheses of interest are \emph{non-sharp}; a primary goal in this paper is to develop procedures that are both theoretically valid (Section~\ref{section:valid}) and 
computationally tractable (Section~\ref{sec:perm_designs}) for such hypotheses.

To illustrate the distinction between sharp and non-sharp null hypotheses, first
let $\Zobs$, $W^{\obs} = w(\Zobs)$ and 
$\Yobs = Y(Z^{\obs})$ be, respectively, the observed assignment, exposure, and outcome vectors. 
We say a null hypothesis is \emph{sharp} if, given the null and the observed data, the potential outcomes $\Yw_i(W_i)$ are imputable for all possible exposures $W_i \in \Hset$, for all units $i \in \Ib$. 

First, consider the global null hypothesis:
\begin{equation}\label{eq:type-0}
	H_0: \Yw_i(\h_1) = \Yw_i(\h_2) = \cdots = \Yw_i(\h_m) \text{ for all } i \in \Ib .
\end{equation}
The null hypothesis in \eqref{eq:type-0} is sharp.
As we show in Section~\ref{section:rand-sharp}, we can test this hypothesis using a standard FRT; \citet[][Section 7.1]{li2018randomization} briefly consider this approach as well.
This global sharp null is analogous to the omnibus null hypothesis in a classical analysis of variance \citep{ding2018randomization} and is a useful starting point for analyses: if there is no evidence of any effect at all, then further analyses are likely less interesting. See \citet[Ch. 15]{lehmann2006testing}. 

At the same time, many substantively interesting causal hypotheses for peer effects are not sharp. 
One important example is the pairwise null hypothesis of the type:
\begin{equation}\label{eq:type-1}
	H_0^{\h_1,\h_2}: \Yw_i(\h_1) = \Yw_i(\h_2) \text{ for all } i \in \Ib,
\end{equation}
where $\h_1,\h_2 \in \Hset$. 
To illustrate, Example \ref{ex:cai} has three possible exposures $\Hset = \{0,1,2\}$, and the sharp null hypothesis of \eqref{eq:type-0} can be written as:
$
		H_0: \Yw_i(0) = \Yw_i(1) = \Yw_i(2) \text{ for all } i\in \Ib.
$
This contains strictly more information about the missing potential outcomes than a pairwise null hypothesis \eqref{eq:type-1}, such as
	$
		H_0^{1,2}: \Yw_i(1) = \Yw_i(2) \text{ for all } i\in \Ib.
$	
Substantively, the global sharp null hypothesis assumes that changing the number of peer CEOs from large firms has no effect whatsoever on a firm's revenue. By contrast, the pairwise non-sharp null hypothesis instead imposes that there is no impact on firm revenue of having one versus two peer CEOs from large firms, without imposing any restrictions on revenue in the absence of any peer CEOs from large firms.
Thus, the ability to test pairwise 
null hypotheses is critical for learning more from the experiment than the initial conclusion that the experiment indeed had some effect somewhere.

Finally, we are often interested in null hypotheses for the subset of units with a given attribute $A_i = a$. As we discuss in our applications below, we often believe that the exposure 
will have differential effects depending on an individual's own attribute. 
Specifically, we can modify both \eqref{eq:type-0} and~\eqref{eq:type-1} to only consider units with $A_i = a$:
\begin{equation}\label{eq:type-00}
	H_0(a): \Yw_i(\h_1) = \Yw_i(\h_2) = \cdots = \Yw_i(\h_m)
 \text{ for all }i\in \Ib \text{ such that } A_i = a
\end{equation}
and
\begin{equation}\label{eq:type-2}
	H_0^{\h_1,\h_2}(a): \Yw_i(\h_1) = \Yw_i(\h_2)
 \text{ for all } i\in \Ib \text{ such that } A_i = a.
\end{equation}
The results below immediately carry over to these subgroup null hypotheses by conditioning on the set of units 
with $A_i = a$. We therefore focus on the simpler null hypotheses of~\eqref{eq:type-0} and 
\eqref{eq:type-1}, returning to subgroup null hypotheses in Section~\ref{section:applications}.

We note that this framework does not require formally specifying an alternative hypothesis; see \citet{athey2018exact} for a discussion in the context of randomization tests under network interference.
In our applications, the choice of the test statistic is motivated by having power against two-sided alternative hypotheses on coefficients from a linear regression model, such as the coefficient on $W$ in the regression of $Y$ on $A$ and $W$.

\subsection{Focal units}
\label{sec:focal_units}

An important technical device for randomization tests under interference is restricting the test to a subset of units known as \emph{focal units} \citep{aronow2012general, athey2018exact, basse2019rand, puelz2022graph}. The intuition behind the approach is that although $H_0^{\h_1,\h_2}$ is not sharp, 
we can ``make the null hypothesis sharp'' by conditioning on a set of focal units 
that are informative about the treatment exposures of interest.

In the context of group formation experiments, we use a binary variable $U_i$ to indicate whether unit $i$ 
is selected as a focal unit; e.g., 
to test $H_0^{\h_1,\h_2}$ we can define $U$ as follows:
\begin{equation}\label{eq:focals-1}
U = u(Z) = (U_1, \ldots, U_N) \in\{0,1\}^N,~\text{with}~U_i=1~\text{if and only if}~w_i(Z)  \in\{ \h_1, \h_2\}.
\end{equation}
That is, we select as focal units the set of units that receive either exposure $\h_1$ or exposure $\h_2$ under assignment $Z$. 
The realized set of focal units, $U^\obs = u(Z^{\obs})$, therefore denotes the set of all units with \emph{observed} exposure $\h_1$ or $\h_2$, the null exposures of interest.
To illustrate, for testing the pairwise null hypothesis $H_0^{1,2}$ in Example \ref{ex:cai}, the focal units are all CEOs who have $W_i^{\obs} = 1$ or $W_i^{\obs} = 2$ peer CEOs from large firms. 
So long as we restrict testing to this subset of units --- and under some restrictions on the possible assignment vectors --- the null hypothesis $H_0^{\h_1,\h_2}$ behaves like a sharp null hypothesis.
\citet{basse2019rand} build on this intuition and develop a valid conditional testing procedure. We adapt this to the group formation design setting in Section \ref{section:rand-non-sharp} below.

\subsection{Toy example and sketch of key ideas}
\label{section:challenges}

Before turning to the theoretical results, we first illustrate the key challenges through a toy example, shown in Figure~\ref{fig:example}.
For this example, individuals possess a binary attribute, represented by squares ($A_i=1$) and circles ($A_i=0$), and are assigned to one of three dorm rooms, one with size 3 (Room I, a ``triple'') and two with size 2 (Rooms II and III, ``doubles''), shown as large rectangles.\footnote{The sizes of the rooms themselves are not central here, and merely restrict the set of possible exposures. We also mean no disrespect to any of our former roommates, several of whom could be described as ``squares.''} 
Rooms are assigned via a \emph{completely randomized group formation design} (see Section \ref{section:crd}), which means that the sizes of the three rooms are fixed, but that the number of square roommates in each room can vary.
Here the exposure mapping is the number of roommates with $A_j = 1$ as defined in \eqref{eq:exposure-2}, so that $\Hset = \{0, 1, 2\}$.
Figure \ref{fig:example} shows the realized assignment $Z^{\obs}$ and induced exposure $W^{\obs}$. 

In this toy example, we are interested in testing two null hypotheses. 
First, the global sharp null hypothesis is that individuals' outcomes are the same regardless of the number of ``square'' roommates. Written in terms of unit-level outcomes, this is $H_0: \Yw_i(0) = \Yw_i(1) = \Yw_i(2)$ for all $i \in \Ib$. 
Second, a non-sharp, pairwise null hypothesis is whether there is an effect of having zero versus one ``square'' roommate, $H_0^{0,1}: \Yw_i(0) = \Yw_i(1)$ for all $i \in \Ib$.

\begin{figure}[t!]
\centering
\includegraphics[scale=0.46]{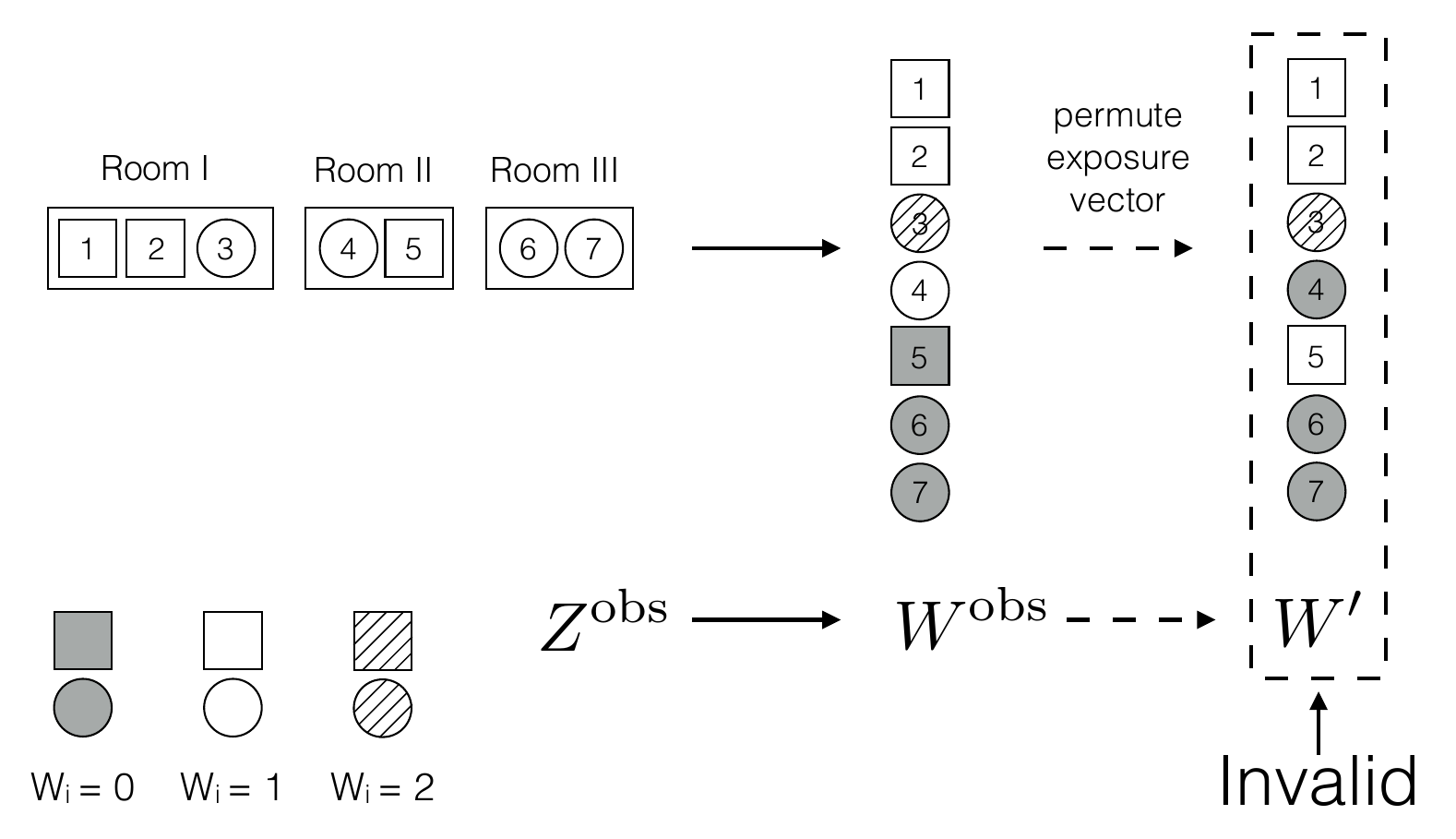}
\caption{Example of a group formation design.
Squares represent units with attribute $A_i=1$ and circles units with attribute $A_i=0$. Units with 
exposure $W_i=0$ (i.e., zero ``square" roommates) are shaded grey; units with exposure $W_i=1$ have no color; and units with exposure $W_i = 2$ (only unit 3) have a patterned background. 
}
\label{fig:example}
\end{figure}

\paragraph{Naive permutation tests can fail.} 
One seemingly natural starting place for testing the global sharp null $H_0$ is a permutation test based on permuting the assignment vector $W^{\obs}$. The right-hand column in Figure~\ref{fig:example} shows one possible permutation $W'$, which switches the exposures of units 4 and 5. This permutation, however, is incompatible with the group formation design; that is, there are no assignments $Z'$ such that, $w(Z') = W'$.\footnote{To see this, note that under $W'$ units 1, 2, and 5 would all need to have exactly one ``square'' roommate. But this isn't possible given the room configuration and group sizes.} 
Thus, naively permuting the exposure vector does not lead to a valid test here.

For the non-sharp null hypothesis, $H_0^{0,1}$, following our setup in Section \ref{sec:focal_units} above, the focal units are the set of units with observed exposure $W_i^{\obs} = 0$ or $W_i^{\obs} = 1$; unit 3, with $W_i^{\obs} = 2$, is the only unit excluded from the focal set. In the context of other interference settings \citep[e.g.,][]{athey2018exact, basse2019rand}, restricting to the focal units would be enough to ensure validity. 
Unfortunately, that is not enough here and naively permuting the exposures again fails --- the permutation $W'$ shown in Figure \ref{fig:example} is, in fact, restricted to focal units but is nonetheless inconsistent with the randomized group formation design.

\paragraph{Randomization tests based on draws from the assignment distribution are valid but computationally prohibitive.} Where permuting $W^{\obs}$ can fail, we can instead re-draw room assignments directly, $Z' \sim \pr(Z')$, and compute the induced exposures for each assignment, $W' = w(Z')$. This leads to valid, direct randomization-based tests for the \emph{sharp} global null hypothesis, $H_0$, though these are not themselves permutation tests. 

However, extending this to non-sharp null hypotheses like $H_0^{0,1}$ is challenging since we must preserve the set of focal units that we condition on. In particular, we must draw from the conditional distribution of room assignments that {\em preserves} the set of focal units; that is, all room assignments such that unit 3 always has observed exposure $W_i^{\obs} = 2$. While we could enumerate all such room assignments in this toy example, this is infeasible in general.

\paragraph{Permutation tests stratified by attribute are valid and tractable for both sharp and non-sharp null hypotheses.} Somewhat remarkably, we can generate valid, computationally tractable randomization tests for both sharp and non-sharp null hypotheses by simply stratifying the permutations based on attribute $A$. 
For the global sharp null, $H_0$, this is among all 7 units; for the pairwise non-sharp null, $H_0^{1,2}$, this is among the 6 focal units. 
In Figure \ref{fig:example}, this is the set of permutations that separately permute the exposures for circles and squares; $W'$ failed because it swapped the exposures of units with different attributes. 

While the final procedure is straightforward, to show that this restricted permutation procedure is valid we must first develop appropriate notions of symmetry and generalize existing group-theoretic results for permutation tests in randomized trials. We turn to this next.

\section{Valid tests  in arbitrary group formation designs}
\label{section:valid}

In this section, we introduce conceptually general --- albeit possibly infeasible --- procedures for constructing valid  tests for sharp and non-sharp null hypotheses for arbitrary group formation designs. 
For sharp null hypotheses, the procedure is a straightforward application of the standard  FRT to our setting. 
For non-sharp null hypotheses, however, the procedure requires greater care to ensure validity.
We turn to constructing feasible randomization tests in the next section.


\subsection{Randomization test for the sharp null}
\label{section:rand-sharp}

We start with a brief review of the classical FRT for sharp null hypotheses~\citep{fisher1935, lehmann2006testing, imbens2015causal}, as a stepping stone to the 
more challenging non-sharp null hypotheses discussed in Section~\ref{section:rand-non-sharp}. 
Consider a test statistic $T(z; Y)$ as a function of the observed treatment and outcome vectors; any choice will lead to a valid test, but certain statistics will lead to more power. 
One reasonable choice, for example, would be the coefficient of $W$ in the regression of $Y$ on $(W, A)$ and other covariates; see also Section~\ref{section:cai} for an applied example.
We can test the sharp null hypothesis $H_0$ with Procedure~\ref{proc:sharp} below.
\begin{procedure}\label{proc:sharp}
Consider observed assignment $\Zobs \sim \pr (\Zobs)$.
\begin{enumerate}
	\item Observe outcomes, $\Yobs = Y(\Zobs)$.
	\item Compute test statistic $T^{\obs} = T(\Zobs; \Yobs)$.
	\item For $Z' \sim \pr (Z')$, let $T' = T(Z'; \Yobs)$ and define
	$ 
		\mathrm{pval}(\Zobs) = \pr (T' \geq T^{\obs}) , 
$
	where $T^{\obs}$ is fixed and the randomization distribution is with respect to $\pr (Z')$.
	\end{enumerate}
\end{procedure}
This procedure is computationally straightforward if the analyst has access to the assignment mechanism $\pr (Z)$, which is necessary for Step~3. 

\begin{proposition}\label{th:sharp}
	The p-value obtained in Procedure~\ref{proc:sharp} is valid, in the sense that if $H_0$ is true, then
$
		\pr \{\mathrm{pval}(\Zobs) \leq \alpha\} \leq \alpha
$
for any 	$\alpha \in [0,1]$. 
\end{proposition}
	In general, it is difficult to compute $\mathrm{pval}(\Zobs)$ exactly, and we must rely on Monte Carlo approximation. This can be done by replacing the third step above by:
	\begin{enumerate}
	\setcounter{enumi}{2}
	\item For $r = 1, \ldots R$, draw $Z^{(r)} \sim \pr (Z^{(r)})$ and compute $T^{(r)} = T(Z^{(r)}; \Yobs)$. Then compute the approximation 
	$
\mathrm{pval}(\Zobs) \approx R^{-1} \sum_{r=1}^R \iv(T^{(r)} \geq T^{\obs}) .
$
	\end{enumerate}

  In practice, the test statistic $T$ used in Procedure~\ref{proc:sharp} is
  chosen to depend on $Z$ only through the exposures $W = w(Z)$. 
 Following our convention in Section~\ref{section:preliminaries}, we can re-write this test statistic as $T(Z; \Yobs) = \Tw(W; \Yobs)$.
  Procedure~\ref{proc:sharp} can then be reformulated as:
  \begin{manualproc}{1b}[special case]\label{proc:sharp-W}
    Consider observed assignment $\Zobs \sim \pr (\Zobs)$.
\begin{enumerate}
\item Observe outcomes, $\Yobs = \Yw(\Wobs)$. 
\item Compute test statistic $T^{\obs} = \Tw(\Wobs; \Yobs)$.
\item For $W' \sim \pr (W')$, let $T' = \Tw(W'; \Yobs)$ and define
$ 
    \mathrm{pval}(\Zobs) = \pr (T' \geq T^{\obs}) , 
$ 
  where $T^{\obs}$ is fixed and the randomization distribution is with respect to $\pr (W')$.
\end{enumerate}
\end{manualproc}

The distribution $\pr(W')$ used above is directly induced by $\pr(Z')$, as
$\pr(W') = \pr\{ w(Z') \} $, and the validity of Procedure~\ref{proc:sharp-W} follows from that
of Procedure~\ref{proc:sharp}, as established by Proposition~\ref{th:sharp}.


\subsection{Randomization tests for non-sharp nulls}
\label{section:rand-non-sharp}

We now turn to the more challenging problem of testing non-sharp pairwise hypotheses such as $H_0^{\h_1,\h_2}$. In general, Procedure~\ref{proc:sharp} can only be valid if the test statistic is imputable 
under $H_0$~\citep{basse2019rand}; that is, $T(Z; Y(Z)) = T(Z; Y^{\obs})$ under $H_0$, for 
all $Z$ for which $ \pr (Z) > 0$.  This property holds because $H_0$ is sharp, which implies that 
$Y(Z) = Y^{\obs}$ under $H_0$. In contrast, pairwise null hypotheses like $H_0^{\h_1,\h_2}$ are not 
sharp, and the FRT methodology does not apply directly.

As we discuss in Section \ref{sec:focal_units} above, we can ``make the null hypothesis sharp'' by restricting the test to the set of focal units, $U^\obs = u(Z^{\obs})$ \citep{aronow2012general, athey2018exact, basse2019rand}. In particular, we use \citet{basse2019rand}'s formulation of conditional tests that guarantee that the resulting test statistics are imputable.
Applying this approach to the peer effects setting requires two changes to Procedure~\ref{proc:sharp}.
First, we need to resample assignments (Step 3 of Procedure~\ref{proc:sharp}) with 
respect to the conditional distribution of treatment assignment,
\begin{equation}\label{eq:cond_mech}
\pr \{ Z' \mid u(Z') = U^{\obs}\} \propto \iv\{u(Z') = U^{\obs}\}\pr (Z') ,
\end{equation}
rather than with respect to the unconditional distribution. In the terminology of \citet{basse2019rand}, 
$U^{\obs}$ is the conditioning event of the test, and its (degenerate) conditional distribution $\pr(U \mid Z) = \iv\{ u(Z) = U \} $ 
is the conditioning mechanism. 

Second, to ensure that the potential outcomes used by the test are imputable, we need to restrict the test statistic to the 
units in the focal set; we denote this new test statistic as $T(z; Y, U)$. For simplicity, we use the restricted difference in means between focal units who are exposed to $\h_1$ and those who are exposed to $\h_2$: 
\begin{equation}\label{eq:test-stat}
	T(z; Y, U) =  \frac{\sum_{i=1}^N \iv\{U_i=1, w_i(z) = \h_2\} Y_i }{\sum_{i=1}^N \iv\{U_i=1, w_i(z) = \h_2\}}
 - \frac{\sum_{i=1}^N \iv\{U_i=1, w_i(z) = \h_1\} Y_i }{\sum_{i=1}^N \iv\{U_i=1, w_i(z) = \h_1\}}.
\end{equation}
The following procedure leads to a valid test of the pairwise non-sharp hypothesis $H_0^{\h_1,\h_2}$.
\begin{procedure} \label{proc:general} Consider observed assignment $\Zobs \sim \pr (\Zobs)$. 
	\begin{enumerate}
		\item Observe outcomes, $\Yobs = Y(\Zobs)$. 
		\item Let $U^{\obs} = u(\Zobs)$ and compute $\Tobs = T(\Zobs; \Yobs, U^{\obs})$ 
  from~\eqref{eq:test-stat}.
		\item For $Z' \sim \pr (Z' \mid U^{\obs})$, let $T' = T(Z'; \Yobs, U^{\obs})$ and define the p-value as 
		$ 
		\mathrm{pval}(\Zobs) = \pr (T' \geq \Tobs \mid U^{\obs}),
		$ 
		where $\Tobs$ is fixed and the randomization distribution is with respect to $\pr (Z' \mid U^{\obs})$ as defined in~\eqref{eq:cond_mech}
	\end{enumerate}
\end{procedure}
As in Section~\ref{section:rand-sharp}, we generally consider test statistics
that depend on $Z$ only through the exposure vector $W = w(Z)$. In addition, notice that
the focal indicator $U = u(Z)$ in \eqref{eq:focals-1} also depends on $Z$
only through $W$.
Following our convention in Section~\ref{section:preliminaries}, this allows us to redefine the focal 
indicator as $U = u(Z) = \uw(W)$, and rewrite Procedure~\ref{proc:general} as follows:

\begin{manualproc}{2b}[special case]\label{proc:general-W}
Consider observed assignment $\Zobs \sim \pr (\Zobs)$. 
\begin{enumerate}
\item Observe outcomes, $\Yobs = \Yw(\Wobs)$. 
\item Compute $U^\obs = \uw(\Wobs)$ and  $\Tobs = \Tw(\Wobs; \Yobs, U^{\obs})$.
\item For $W' \sim \pr (W' \mid U^{\obs})$, let $T' = \Tw(W'; \Yobs, U^{\obs})$ and define the p-value as
$ 
    \mathrm{pval}(\Zobs) = \pr (T' \geq \Tobs),
$ 
  where $\Tobs$ is fixed and the randomization distribution is with respect to
  $\pr (W' \mid U^{\obs})$. Note again that the distribution $\pr(W' \mid U^{\obs})$ is induced by that of
$\pr(Z' \mid u(Z')=U^{\obs})$.
\end{enumerate}
\end{manualproc}

\begin{proposition}\label{th:general}
  Procedure~\ref{proc:general} and its special case,
  Procedure~\ref{proc:general-W}, lead to valid p-values conditionally and
  marginally for $H_0^{\h_1,\h_2}$. That is, if $H_0^{\h_1,\h_2}$ is true then $
\pr \{ \mathrm{pval}(\Zobs) \leq \alpha  \mid U^{\obs}\} \leq \alpha $
  for any $U^{\obs}$  and any $\alpha \in [0,1]$,
and thus
$\pr \{ \mathrm{pval}(\Zobs) \leq \alpha   \} \leq \alpha$ as well.
\end{proposition}

The proof for Proposition~\ref{th:general} is a direct application of Theorem 1 of~\citet{basse2019rand}. For the rest of this paper, we only consider test statistics that depend on $Z$ through $W = w(Z)$ alone. Therefore, all the statements in subsequent  sections will be made in terms of Procedures~\ref{proc:sharp-W}
  and~\ref{proc:general-W} instead of Procedures~\ref{proc:sharp} and~\ref{proc:general}. 

The conditional randomization tests described in this section
differ from standard conditional tests in several important ways. First, the
goal of standard conditional tests is typically to make the test more powerful \citep{lehmann2006testing, hennessy2016conditional}, rather than to
ensure validity. The conditioning in Procedures~\ref{proc:general}  and~\ref{proc:general-W}, by contrast, is necessary to ensure that the
test is valid.
Second, the procedure depends strongly on the non-sharp null hypothesis being
tested. 
Indeed, conditional randomization tests can only test some non-sharp null hypotheses, such as $H_0^{\h_1,\h_2}$, which typically dictate the conditioning mechanism.

\paragraph{Computational challenges with testing non-sharp nulls.}

The key challenge for testing non-sharp null hypotheses is that the procedures outlined above are computationally intractable in realistic 
 settings. Indeed, while we can easily draw samples from the unconditional 
distribution $\pr(W)$ through $w(Z)$, where $Z\sim\pr(Z)$, Step~3 of Procedure~\ref{proc:general-W} requires draws 
from the unwieldy conditional distribution $\pr(W \mid U^\obs)$.
 Our main proposal in the next section directly addresses this computational issue.



\section{Using design symmetry to construct computationally tractable permutation tests}
\label{sec:perm_designs}
\label{section:overview}

This section motivates the use of permutation tests for a broad class of randomized group formation experiments. 
To do so, we show that certain designs can lead to computationally tractable conditional distributions $\pr(W \mid  U)$, which are crucial in the randomization tests discussed above.
This section relies on results from algebraic group theory; readers interested in the concrete consequences of these results on the design of randomization tests in our setting may skip ahead to Section~\ref{section:practical}.

\newcommand{\SN}{\mathbb{S}_N}
\newcommand{\SAU}{\mathbb{S}_{A,U}}

 \subsection{Equivariant maps and stabilizers} \label{section:group-theory}

 This subsection introduces three key algebraic concepts for our main theoretical result.
Let $\SN$ be the symmetric group containing all permutations of $N$ elements; 
i.e., bijections of $\{1, \ldots, N\}$ onto itself.  For any permutation $\pi \in \SN$ and a real-valued $N$-length vector $X \in\mathbb{X}\subseteq\mathbb{R}^N$, let $\pi X = (X_{\pi^{-1}(i)})_{i=1}^N$ be the vector obtained by permuting the indices of $X$ according to $\pi$.
\begin{definition}[Stabilizer]
$\mathbb{X}$ is closed under $\SN$ in the sense that $\pi X \in \mathbb{X}$ for all $\pi \in \SN$ and  $X \in \mathbb{X}$. 
Fix $X\in\mathbb{X}$.
The set
  $\SN(X) = \{\pi \in \SN: \pi X = X\}$
  also forms a group and is called the stabilizer of $X$ in $\SN$.
\end{definition}
A stabilizer $\SN(X)$ captures all possible ways of permuting
$X$ without changing $X$. For instance, if $X$ is a binary vector, then a
permutation $\pi \in \SN(X)$ separately permutes elements with $X_i = 0$ and $X_i=1$, respectively. This formalizes the argument we sketched out in Section~\ref{section:challenges}: the operations that ``permute units with the same attribute'' are precisely the elements of $\SN(A)$, the stabilizer of the attribute vector $A = (A_i)_{i=1}^N$ in the symmetric group.

\begin{definition}[Orbits and Partitions]
Fix a subgroup of the symmetric group $\Pi\subseteq \SN$.
Fix $X\in\mathbb{X}$, where $\mathbb{X}$ is closed under $\Pi$.
  Then, the set
  $\{\pi X  : \pi \in \Pi\}$
  is called the orbit of $X$ with respect to $\Pi$.
  These orbits define a unique partition of $\mathbb{X}$,
  denoted by $\mathcal{O}(\mathbb{X}; \Pi)$.
\end{definition}
Thus, an orbit is a collection of vectors that are permuted versions of one another. A key property of orbits is that they partition the set that the permutations act upon. This is important in our application because
our permutation test on $W$ essentially conditions on an orbit, and we would like 
the symmetries of our design $\pr(L)$ to be propagated to the conditional distribution of $L$ given an orbit. The 
final property that guarantees such symmetry propagation is equivariance.
\begin{definition}[Equivariant maps]
Fix a subgroup of the symmetric group $\Pi\subseteq \SN$. 
Sets $\mathbb{X}$ and $\mathbb{X}'$ are closed under  $\Pi$ in the sense that $\pi X\in \mathbb{X}$ and $\pi X' \in \mathbb{X}'$ for all $X\in \mathbb{X}, X' \in \mathbb{X}'$ and $\pi \in \Pi$. A function $f: \mathbb{X} \to \mathbb{X}'$ is equivariant with respect to $\Pi$ if
  $$
  f(\pi X) = \pi f(X),~\text{for all}~X\in\mathbb{X},~\pi\in\Pi.
  $$
\end{definition}
By definition, equivariant maps preserve a symmetry from their domain to their target set. 
This concept is crucial for our main theoretical result, which we turn to next.

\subsection{Main result: Sufficient conditions for valid permutation tests on  exposures}
\label{section:propagating}

We now state our main theoretical result, which establishes that if the exposure function, $\wl(\cdot)$, and the focal unit selection function, $\ul(\cdot)$, are equivariant with respect to a particular permutation subgroup, 
then the treatment exposure $W$, is uniformly distributed within an orbit defined by that subgroup. 

\begin{theorem}\label{thm1}
 \TheoremOne
\end{theorem}

Theorem~\ref{thm1} formalizes the intuition
behind the example in Section~\ref{section:challenges}: under the conditions of the theorem, we can implement Procedure~\ref{proc:general-W} by directly permuting the exposures of {\em only} the focal units, and making sure that these permutations are stratified with respect to the attribute value; the space of these permutations is exactly $\SAU$. 
The sharp null of Procedure~\ref{proc:sharp-W} is a special case of this result by defining $\ul(L) = \ones_N$, i.e., by selecting all units to be focals. 
In this special case, $\SAU = \SN(A)$ and so we can directly permute the entire exposure vector, $W$, across units with the same attribute value.

All three conditions in Theorem~\ref{thm1} are intuitive and testable in practice.
Condition (a) expresses a {\em design symmetry} condition. This depends on the experimental design, and will generally be satisfied for a permutation group that is larger than $\SAU$, such as in the stratified and completely randomized designs we consider in the next section. 
In particular, the design symmetry condition holds for both our applications.
For instance, in~\citet{cai2017interfirm}, the design is invariant to permutations between firms of the same size and industry in the same subregion~(i.e., the  attribute $A$ is a vector of length 3); we discuss this condition more in Section~\ref{section:applications}.

Condition (b) depends on the definition of the exposures, and is part of the analysis rather than the design. This condition posits that, for two units with the same attribute $A$ and focal status $U$, swapping the group label assignments also swaps their exposures;
Condition (b) does not require the exclusion restriction in Assumption \ref{asst:properly-specified}.
Finally, Condition (c) is also under the analyst's control and requires that swapping group label assignments for two units also swaps their selection as focal units.

We note that Theorem~\ref{thm1} is more general than the specific group formation design settings we consider in this paper. 
In particular, our definition of the exposure function $\wl(\cdot)$ in Eq.~\eqref{eq:exposure} satisfies Condition (b), and our definition of the focal selection function $\ul(\cdot)$ in Eq.~\eqref{eq:focals-1} satisfies Condition (c). 
In fact, Condition (c) holds more generally whenever focal selection depends on whether the observed exposure belongs to a predefined set. 
We summarize these results in the following lemma; see~Appendix~\ref{proof:conditions} for its proof.
\begin{lemma}
    \label{lemma::conditions-b-c}
     Conditions (b)-(c) of Theorem~\ref{thm1} hold with the definitions in Eqs. \eqref{eq:exposure} and \eqref{eq:focals-1}.
\end{lemma}
Since Conditions (b) and (c) hold in our setting, we will only check the design symmetry in Condition~(a) going forward.

Finally, as a technical note, Theorem~\ref{thm1} contributes to the existing theory of randomization tests by 
providing sufficient conditions under which symmetry in distribution of a random variable 
implies symmetry in distribution to a {\em function} of that variable.
In our context, while the standard theory of randomization tests ~\citep{lehmann2006testing} could be applied on hypotheses in the space of labels ($L$), it is not directly applicable in the space of exposures, $W=\wl(L)$. This is because $W$ is not generally invariant to permutations even when $L$ is. 
The toy example in Section~\ref{section:challenges} illustrated this point by considering permutations of the exposure vector that were inconsistent with the experimental design. Theorem~\ref{thm1} delivers conditions under which $W$ maintains a permutation symmetry like $L$. Crucially, the theorem also characterizes the permutation subgroup ($\SAU$) for which such symmetry propagation is possible.

\section{Permutation tests in two group formation designs} 
\label{section:practical}
We now apply the theory of the previous section
in practice. 
We consider two designs, the stratified randomized design and completely randomized design, and show that these designs have the required symmetries for permutation tests on exposures.

\subsection{Stratified randomized design}
\label{sec:sr_designs}

The stratified randomized design is an important special case of group formation design that satisfies the design symmetry condition in Theorem~\ref{thm1}(a).
Specifically, we consider designs that, separately for each level of attribute $A$, assign $K$ group-labels to $N$ units completely at random.
In a trivial setting with a binary attribute and two individuals per group, this design randomly assigns one individual of each type to each group.
\begin{definition}[Stratified randomized design]\label{def:sr_designs}
	Consider a distribution of group labels, $\pr (L)$, that assigns equal probability to all vectors
	$L$ such that for every attribute $a \in \Aset$ and every group-label $k \in \{1, \ldots, K\}$, the number of units with attribute 
	$A_i = a$ assigned to group-label $k$ is equal to  a fixed integer $n_{a,k}$. The design $\pr (Z)$ induced by such $\pr (L)$ is called a stratified randomized group formation design, denoted by $\SR(\nak_A)$, where $\nak_A   =   (n_{a,k})$ satisfies the constraint that 
$\sum_{k=1}^K n_{a,k} = |\{i \in \Ib: A_i = a\}|$. 
\end{definition}

The stratified randomized design  generalizes the design in   \citet[][Section~2.4.2]{li2018randomization} by allowing the group sizes to vary.
As an illustration, Figure~\ref{fig:SR_illustration} shows all possible assignments for two stratified randomized 
designs in a setting in which we allocate students with a binary attribute to their dorm rooms. The design on the left is $\SR(\nak_A)$ with $(n_{0,1}, n_{0,2})  = (1,2)$, meaning that there is one unit with attribute $A_i = 0$ assigned to room $1$, and two to room $2$; and $(n_{1,1}, n_{1,2}) = (2, 0)$, meaning that that there are two units with 
attribute $A_i = 1$ assigned to room 1, and no unit assigned to room $2$. The design on the right is
$\SR(\nak_A ' )$ 
with $(n_{0,1} ' , n_{0,2} ' ) = (2,1)$ and $(n_{1,1} ' , n_{1,2} ')  = (1,1)$.

\begin{figure}[bt]
	\center
	\includegraphics[scale=0.7]{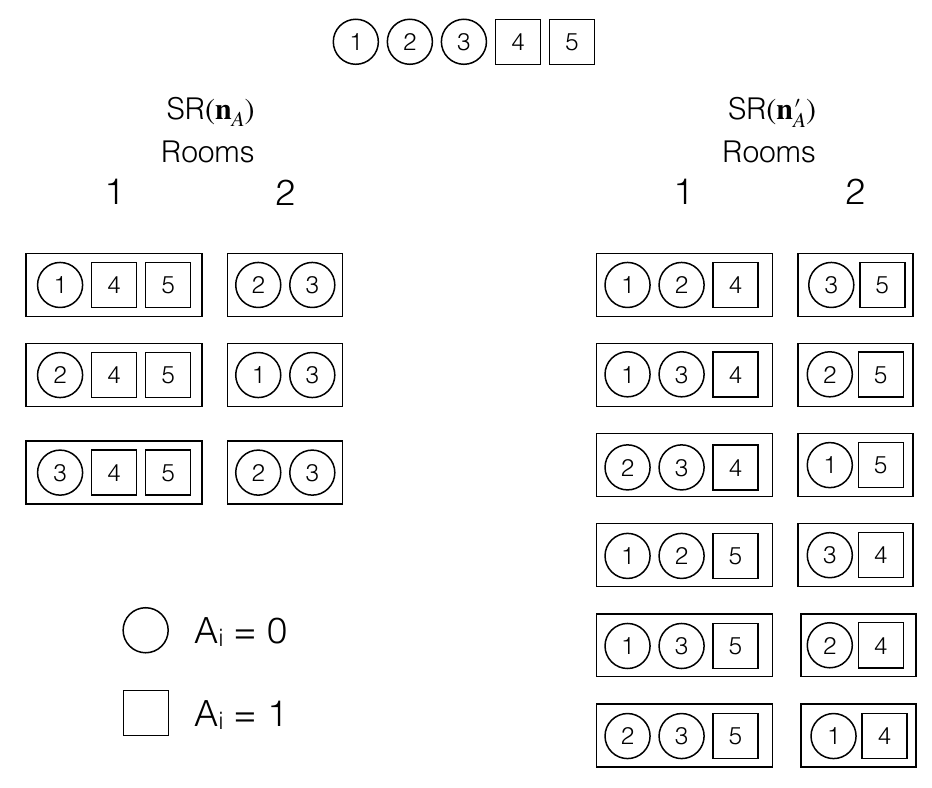}
	\caption{
	Example of supports for two latent distributions $\pr (L)$ inducing two stratified randomized experiments. Both examples have $N=5$ units, $K=2$ rooms labelled 1 and 2, and a binary 
	attribute. Left: $(n_{0,1}, n_{0,2}) = (1,2)$ and $(n_{1,1}, n_{1,2}) = (2, 0)$. 
	Right: $(n'_{0,1}, n'_{0,2}) = (2,1)$ and $(n'_{1,1}, n'_{1,2})= (1,1)$.
	}
	\label{fig:SR_illustration}
\end{figure}

Importantly,  the stratified randomized design satisfies the design symmetry 
condition in Theorem~\ref{thm1}(a) since the number of units assigned to any attribute-label pair remains fixed under any permutation of the labels that stratifies on $A$. See Appendix~\ref{proof:stratified} for the proof.

Our recommended  procedure for testing the sharp null under a stratified design is as follows:
\begin{manualproc}{1c}[Sharp null under the stratified randomized design]\label{proc:sharp-permute}
Consider observed assignment $\Zobs \sim \mathrm{SR}(\nak_A)$ and corresponding 
exposure $\Wobs$.
\begin{enumerate}
\item Observe outcomes, $\Yobs = \Yw(\Wobs)$.
\item Compute $\Tobs = \Tw(\Wobs; \Yobs)$; e.g., as defined in Section~\ref{section:rand-sharp}.
\item For $r=1,\ldots, R$, obtain $W^{(r)}$ via a random permutation of $\Wobs$, stratifying on the attribute $A$, and then compute $T^{(r)} = \Tw(W^{(r)}; \Yobs)$. 
\item Compute the approximate p-value 
$\mathrm{pval}(\Wobs) = R^{-1} \sum_{r=1}^R \iv(T^{(r) }\geq \Tobs)$.
\end{enumerate}
\end{manualproc}

In Step 3 above, we randomly permute $\Wobs$ stratifying on attribute $A$, that is, we randomly permute within each subvector of $\Wobs$ corresponding to a given value of $A$. 
This procedure is identical to how one would analyze a 
stratified completely randomized multi-arm trial in the non-interference 
setting --- with the exposure vector $\Wobs$ being the analog to the
treatment vector in that case \citep[][Chapter 9]{imbens2015causal}. That is, given the data 
$(Y_i, W_i, A_i)_{i=1}^N$, the analyst simply perform a  
complete randomization test stratified on $A$. 

The analogy with the traditional setting extends --- with minor 
modifications --- to testing the non-sharp nulls introduced in Section~\ref{section:rand-non-sharp}. Recall that for Procedure \ref{proc:nonsharp-permute},
the test statistics are restricted to focal units, i.e., $T(z; Y, U)$. 
Our recommended procedure for testing non-sharp nulls under a stratified design is then:
\begin{manualproc}{2c}[Non-sharp nulls under the stratified randomized design]\label{proc:nonsharp-permute}
Consider observed assignment $\Zobs \sim \mathrm{SR}(\nak_A)$ and corresponding 
exposure $\Wobs$.
\begin{enumerate}
\item Observe outcomes, $\Yobs = \Yw(\Wobs)$. 
\item Let $U^{\obs} = u(\Zobs)$ be the focal unit selection as in~\eqref{eq:focals-1}.
\item Compute $\Tobs = \Tw(\Wobs; \Yobs, U^{\obs})$ as in~\eqref{eq:test-stat}.

\item 
For $r=1,\ldots, R$, obtain $W^{(r)}$ via a random permutation of $\Wobs$, restricted only to focal units~($\Uobs_i=1$) and stratifying on the attribute $A$. Compute $T^{(r)} = \Tw(W^{(r)}; \Yobs, \Uobs)$. 
\item Compute the approximate p-value 
$\mathrm{pval}(\Wobs) = R^{-1} \sum_{r=1}^R \iv(T^{(r)} \geq \Tobs)$.
\end{enumerate}
\end{manualproc}
Although less obvious than in the case of Procedure~\ref{proc:sharp-permute}, 
Procedure~\ref{proc:nonsharp-permute} also connects to traditional randomization tests. Indeed, given 
the data $(Y_i, W_i, A_i)_{i=1}^N$, the analyst first subsets the array to 
contain only focal units~($\Uobs_i=1$), and then simply  performs a stratified complete randomization test on this reduced data, stratifying on $A$. 
Interestingly, there is a gap in the literature for randomization 
tests for non-sharp null hypotheses, even in traditional stratified 
randomized experiments without peer effects. Our permutation test applies to the traditional setting as well. Finally, we note that both Procedures~\ref{proc:sharp-permute} and~\ref{proc:nonsharp-permute} are finite-sample exact with a direct
application of Theorem~\ref{thm1}. See Appendix~\ref{proof:stratified} for details.

\subsection{Completely randomized design}
\label{section:crd}

Another common design is the completely randomized design, which fixes the \emph{overall} number of units  that receive each group-label, without stratifying on the attribute. 
Despite this difference, we will show that the completely randomized design can be analyzed exactly like a stratified randomized design by conditioning on the observed attribute-group assignments.

\begin{definition}[Completely randomized design]
	Consider a distribution of group labels, $\pr(L)$, that assigns equal probability to all vectors $L$ such that 
	for every group-label $k \in \{1, \ldots, K\}$, the number of units assigned to group-label $k$ is equal 
	to a fixed integer $n_k$. The design $\pr(Z)$ induced by such $\pr(L)$ is a completely randomized group formation 
	design,  denoted by $\CR(\nak)$, where $\nak = (n_1, \ldots, n_K)$ satisfies $\sum_{k=1}^K n_k = N$.
\end{definition}

The completely randomized design generalizes the design in \citet[][Section~2.4.1]{li2018randomization} 
by allowing the size of the groups to vary. 
Importantly, we can construct a stratified randomized design from a completely 
randomized design by conditioning on the number of units with each level of the attribute in each group. As a result, conditional on $\nak_A$, we can analyze a completely randomized group formation design exactly like a stratified randomized design. 
\begin{corollary}\label{corollary:specific}
	Consider $\pr(Z) \sim \CR(\nak)$. The null hypotheses $H_0$ (resp. $H_0^{\h_1,\h_2}$) 
	can be tested with Procedure~\ref{proc:sharp-permute} (resp. Procedure~\ref{proc:nonsharp-permute}) as if the 
	design were $\SR(\nak_A)$, where $\nak_A$ is the observed number of units with each value of the 
	attribute $A$ assigned to each group. 
\end{corollary}

This connection is important since many designs are not stratified on the attribute of 
interest; e.g., the application we analyze in Section~\ref{section:li} uses a completely randomized 
design rather than a stratified randomization design.
Importantly, conditioning on $\nak_A$ is necessary to ensure the validity of the permutation test even in completely randomized designs. Figure \ref{fig:example} gives an example in which the unconditional permutation test is invalid. 

\begin{remark}[Incorporating additional covariates]\label{section:cov}
All our procedures can be extended to incorporate additional covariates in the design and analysis stages. These strategies will generally increase the power of the test, so long as covariates are predictive of the potential outcomes~\citep{zhao2020covariate}. 
Most immediately, we could stratify both the permutations and the test statistic by an additional discrete covariate, $X$.
We could also consider regression-adjusted test statistics, rather than test statistics based on the raw outcomes \citep{rosenbaum2002observational}.
We could further tailor these models to a particular interference structure; for instance, \citet{athey2018exact} propose a test statistic derived from the linear-in-means model. Importantly, this approach does \emph{not} assume that the linear-in-means model is correct, but rather that this parameterization captures departures from the null hypothesis. 
In Appendix~\ref{appendix:power}, we perform a simulated study to illustrate these points.
\color{black}
\end{remark}

\section{Applications}
\label{section:applications}

We illustrate our approach by re-analyzing two randomized group formation experiments.
The first application is from 
\citet{li2018randomization}, who assess the impact of randomly assigned roommates on student GPA. 
Our conditional testing approach yields results that are consistent with their randomization-based estimate.  
The second application is from \citet{cai2017interfirm}, who conduct a randomized experiment to estimate the effect of social connections on firm performance. 
Our approach complements the results from their regression-based estimates by uncovering interesting heterogeneity in the peer group effect.

\subsection{Random roommate assignment}
\label{section:li}

\citet{li2018randomization} explore the impact of the composition of randomly assigned roommates on student academic performance among students at a top Chinese university.
For ease of exposition, we restrict our analysis to the $N = 156$ male students admitted to the Department of Informatics, the largest department in the original study. 
The attribute of interest is whether students are admitted via a college entrance 
exam ($A_i = 1$), known as \emph{Gaokao}, or via an external recommendation ($A_i = 0$). Students are 
assigned to dorm rooms of size four via complete randomization, as described in Section~\ref{section:crd}; 
that is, the number of students of each background in each room is a random quantity. 

The exposure of interest is the number  of roommates admitted via the entrance exam $w_i(Z) = \sum_{j\in Z_i} A_j$.
We focus on the null hypothesis $H_0^{0,3}: \Yw_i(0) = \Yw_i(3)$ for all $i=1,\ldots, N=156$, that is, a student's 
end-of-year GPA is the same if he is randomly assigned to have zero \emph{Gaokao} roommates versus three 
\emph{Gaokao} roommates. Moreover, following \citet{li2018randomization}, we want to test this null hypothesis 
separately for \emph{Gaokao} and recommendation students, which we denote $H_0^{0,3}(1)$ and 
$H_0^{0,3}(0)$ respectively.
Here, Assumption \ref{asst:properly-specified} states that group formation only affects end-of-year GPA by changing the number of {\em Gaokao} roommates for a student. This excludes, for example, the subject area or sociability of roommates as important mechanisms for group peer effects.
Among 17 students from {\it Gaokao}, 13 have observed exposure $W_i^{\obs}=0$ and 4 have observed exposure $W_i^{\obs}=3$; among 45 students from recommendation, 40 have observed exposure $W_i^{\obs}=0$ and 5  have observed exposure $W_i^{\obs}=3$.
Table~\ref{table:li-pvals} reports the $p$-value, Hodges--Lehmann point estimate,
\footnote{See Appendix \ref{section::hodgeslehmann} for the discussion of the Hodges--Lehmann estimate.}
and test inversion confidence interval for the overall null hypothesis $H_0^{0,3}$ and the subgroup null hypotheses $H_0^{0,3}(1)$ and $H_0^{0,3}(0)$. 

Our results are substantively close to those obtained by \cite{li2018randomization}. 
First, our point estimates 
are identical to those from \citet{li2018randomization} by symmetry. 
Our $p$-values and confidence intervals, however, are more conservative, in the sense of showing weaker evidence against the null.
Specifically, \cite{li2018randomization} find $p$-values $\leq 0.05$ for all three null hypotheses, while we only reject $H_0^{0,3}$ at that level.
One possible explanation for this discrepancy is that, while our $p$-values are exact, \cite{li2018randomization} instead use an asymptotic approximation, which may be unwarranted given the small sample size.
We investigate this more in Appendix~\ref{appendix:li}, where we conduct a calibrated simulation based on this application and show that normal asymptotics can fail severely here.
%

\renewcommand{\arraystretch}{1}
\begin{table}[t!]
	\center
		\caption{$p$-values, Hodges--Lehmann point estimates and $95\%$ confidence intervals for the application of~\cite{li2018randomization}.} 

	\begin{tabular}{|c|c|c|c|c|c|c|}
	\hline
		& $p$-value & estimate & confidence interval \\
		\hline
	$H_0^{0,3}$  & 0.03 & $-0.32$  & $(-0.65, -0.05)$  \\
	$H_0^{0,3}(0)$ & 0.058 & $-0.37$ & $(-0.74, 0.005)$ \\
	$H_0^{0,3}(1)$ & 0.22 & $-0.29$ & $(-0.8, 0.1)$ \\
	\hline
	\end{tabular}
	\label{table:li-pvals}
\end{table}


\subsection{Meeting groups among firm managers}
\label{section:cai}
We now turn to the study from \citet{cai2017interfirm}, in which CEOs of Chinese firms were randomly assigned to meetings where they discussed management practices, with ten managers per group. Groups were encouraged to meet monthly for roughly a year; firms assigned to control did not meet.
The primary outcome of interest is growth in firm sales, defined as the difference in (log) firm sales from endline to baseline.\footnote{\citet{cai2017interfirm}  collected survey data at baseline, midline, and endline. While the authors analyzed the experiment using panel data regression, we side-step the panel structure here by defining the outcome as the difference in log firm sales between endline and baseline. We note, however, that our framework accommodates a wide range of outcomes and test statistics, including those generated by panel regressions.}
\citet{cai2017interfirm} focused on the impact of assigning CEOs to meeting groups versus a business-as-usual control group. 
Here we revisit a secondary analysis in their paper that explores the role of peer composition. In particular, among treated firms, the group formation design was stratified across three attributes: firm sector (manufacturing/service), location (26 subregions), and firm size (small/large).\footnote{Firm size is dichotomized at median employment of the sample of firms in the corresponding subregion, where the authors use the number of employees at baseline as a proxy for the quality of the firm.} 
Using this design, \citet{cai2017interfirm} ``ask whether firms randomized into groups with larger peers grew faster,'' finding evidence in the affirmative. 

We revisit this question using our proposed randomization inference framework, where firm size is the exposure of interest. 
In particular, we focus on the 1,323 firms with non-missing data (on size and revenue) that were randomly assigned to meetings. 
We first consider the global sharp null of any effect of peer size on sales, and 
then highlight a source of peer effect heterogeneity by testing the sharp null within subgroups defined by sector and size.
In the Appendix, we also consider alternative exposure definitions and look at pairwise, non-sharp null hypotheses to further explore this source of heterogeneity. 

\paragraph{Global sharp null hypothesis.}
We start with the global sharp null hypothesis that there is no effect whatsoever of peer size on sales. 
The exposure of interest is $W_i = \frac{1}{|Z_i|}\sum_{j \in Z_i} \text{size}_j$, where $Z_i$ is the set of  peer firms for firm $i$, and $\text{size}_j$ is the log-number of employees in firm $j$ at baseline. Let $\Wdom \subset \mathbb{R}$ be the exposure domain, then  the global sharp null hypothesis is:
\begin{equation}\label{eq:regression}
H_0: \Yw_i(\h) = \Yw_i(\h') \mbox{ for all } i \in \Ib \text{ and }\h, \h'\in \Wdom.
\end{equation}
That is, under $H_0$, the average employee size of firm $i$'s peer group does not affect the firm's revenue. 
As we discuss in Section~\ref{sec:assumptions}, Assumption \ref{asst:properly-specified} plays a critical role in interpreting a rejection of our null hypothesis. In this application, Assumption \ref{asst:properly-specified} states that group formation only affects sales by changing the size of a firm's peer companies. This excludes, for example, the number of other peer firms' \emph{clients} (rather than number of employees) from affecting a firm's own revenue. To check robustness, we explore alternative definitions of the exposure in Appendix \ref{appendix:cai_additional}.

To mirror the analysis in \citet{cai2017interfirm}, we set the test statistic to be the coefficient of $W$ in the following linear regression:\footnote{To aid interpretation, we follow the regression specification in \citet{cai2017interfirm}. However, the randomization inference theory from \citet{li2018randomization} shows that a regression specification that includes the interaction of $A$ and $W$ is also justified by the randomization itself.}
\begin{equation} \label{eq:cai_szeidl_test_stat}
Y_i^{\text{obs}} = \alpha + \beta A_i^* + \tau W_i + \varepsilon_i,
\end{equation}
where $A_i^* = \mathrm{sector}_i \times \mathrm{location}_i \times \mathrm{size}_i$  includes all interactions between firm sector, location, and size for unit $i$. 
We can now employ Procedure~\ref{proc:sharp-W} to test $H_0$, computing a one-sided $p$-value of $p = 0.02$ over 20,000 replications.
Importantly, even if the linear model in Equation \eqref{eq:cai_szeidl_test_stat} is not correctly specified, the randomization test remains finite-sample valid. 

\paragraph{Heterogeneity by firm size and type.}
Since our approach is exact in finite-samples, we can easily restrict our analysis to subsets of firms, here defined by sector and size following \citet{cai2017interfirm}. 
We repeat Procedure~\ref{proc:sharp-W} separately within each subgroup, using the estimated coefficient $\tau$ from Equation \eqref{eq:cai_szeidl_test_stat}, except with the levels of $A^\ast$ restricted to the appropriate subgroup.
The results in Figure \ref{tab1:sharp} show substantial heterogeneity in peer group effects.
In particular, the signal is concentrated entirely among small service firms ($p=0.0015$), and is essentially zero for the other three subgroups.

\begin{figure}[t!]
\centering
\begin{tabular}{cc}
  \includegraphics[scale=0.2]{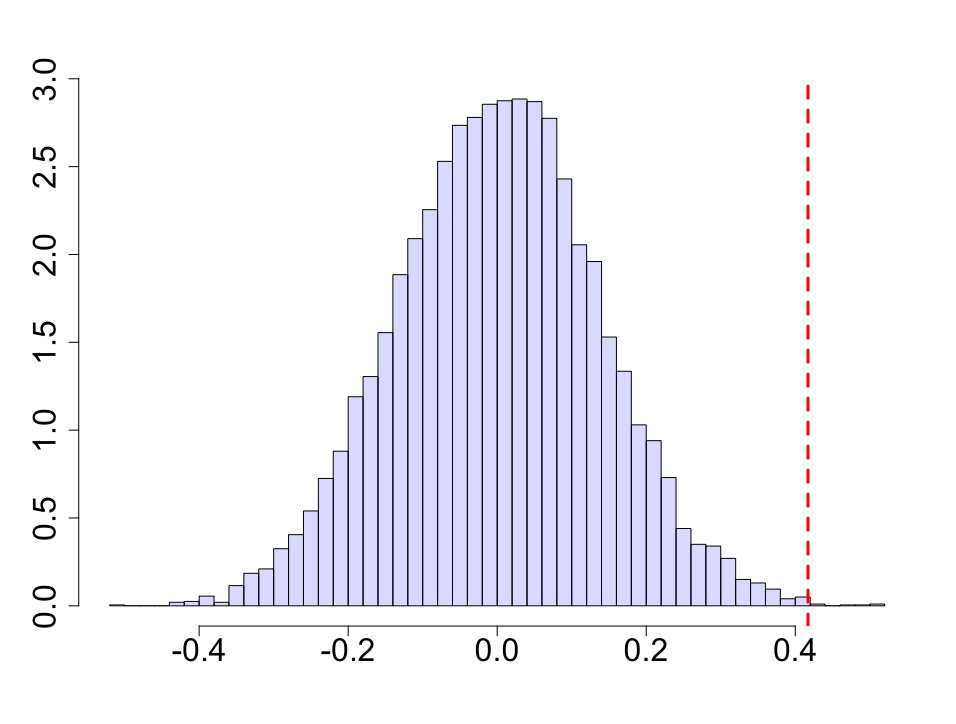}    &  \includegraphics[scale=0.2]{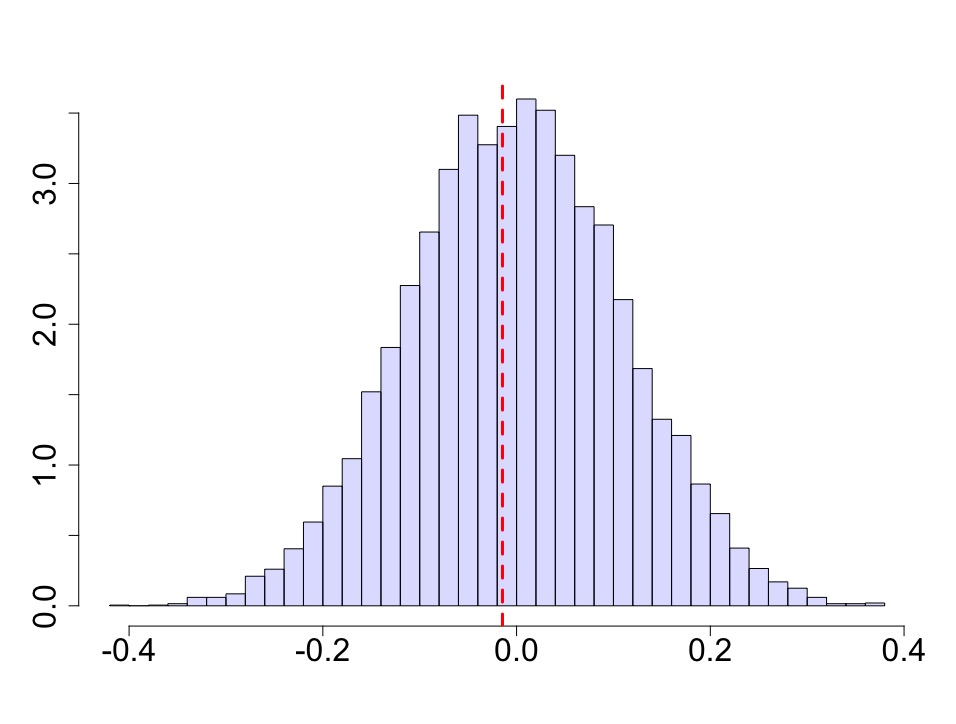}\\
{\small small service firms~($n$=314, $p$=0.0015)} & {\small small manufacturing firms~($n$=349, $p$=0.56)} \\
 \includegraphics[scale=0.2]{art/sharp_01.jpeg}   &  \includegraphics[scale=0.2]{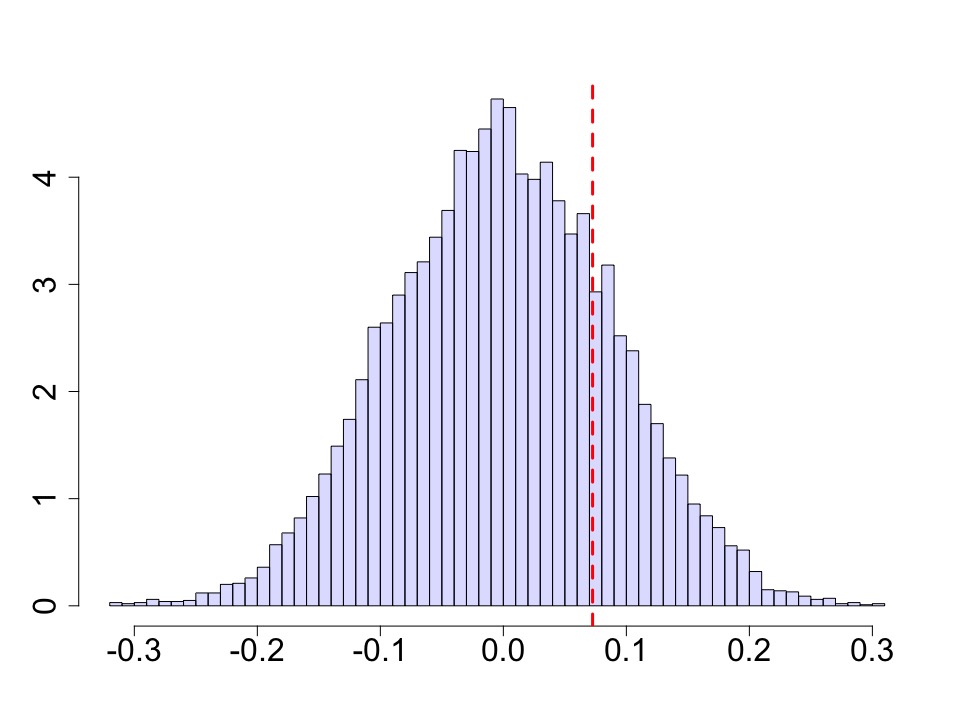}\\
  {\small  large service firms~($n$=333, $p$=0.23) } & {\small large manufacturing firms~($n$=327, $p$=0.21)} \\
\end{tabular}
\caption{Randomization distributions for $H_0$ in~\eqref{eq:regression} within firm subpopulations according to 
size and sector. The dashed line indicates the observed value of the test statistic; 
`$n$' is the subpopulation size, and `$p$' is the one-sided p-value calculated from Procedure~\ref{proc:sharp-W}.
}
\label{tab1:sharp}
\end{figure}

\citet{cai2017interfirm} also explored heterogeneity, albeit only in the ``direct effects" from treatment (i.e., meetings versus no meetings) rather than in peer effects; they find larger firms benefited more from the meetings. 
Our analysis complements this picture by showing that the impact of larger peers was concentrated mainly among small service firms.
We emphasize that the regression specification of~\cite{cai2017interfirm} in~\eqref{eq:regression} cannot easily capture the heterogeneity we show here.
In particular, their regression model needs to include all size-sector-subregion interactions ($\sim$85 in total) dictated by the experimental design in order to identify $\tau$~\citep[see Section III.B in][]{cai2017interfirm}.
These interactions, however, essentially ``wash out" the size-sector interaction effect we observe here.
Thus our randomization-based analysis complements the regression-based analyses and offers new insights. 

Finally, an additional benefit of our analysis is that our $p$-values are exact, which is especially important for subgroups. In Appendix~\ref{appendix:cai_sim}, we highlight this  through a simulation study showing that regression-based tests can be severely distorted in simple but realistic group formation  designs motivated by \citet{cai2017interfirm}.


\section{Discussion}
\label{section:discussion}

We have proposed valid randomization tests for testing peer effects in group formation experiments. 
While a promising first step, there remain several open questions. First, our results motivate new considerations for the design of group formation experiments.  
In particular, arbitrary designs do not necessarily satisfy the sufficient conditions we propose for valid permutation tests. We therefore recommend using the experimental designs like the stratified and completely randomized designs in Section~\ref{section:practical} if researchers want to use our permutation-based tests.

Second, our approach is limited to the setting where units are assigned to groups. However, sometimes the group structure might be more elaborate. For example, we might assign students to classrooms and then separately assign teachers to those classrooms. Alternatively, we might be interested in multiple, possibly overlapping groups; e.g., students being nested within classrooms nested within schools. 
Finally, we have focused entirely on randomized group formation experiments. Randomizing peers, however, may often be infeasible or raise ethical concerns. Thus, extending these ideas to the observational study setting, especially for sensitivity analysis, is a promising avenue for future work.

 \singlespacing
\bibliographystyle{chicago}
\bibliography{refs}

\clearpage 

\appendix
\begin{center}
 \huge  \bf   Online Appendix
\end{center}


\pagenumbering{arabic} 
\renewcommand*{\thepage}{S\arabic{page}}

\setcounter{equation}{0}
\setcounter{section}{0}
\setcounter{figure}{0}
\setcounter{example}{0}
\setcounter{table}{0}
\setcounter{lemma}{0}
\setcounter{theorem}{0}

\renewcommand {\thetable} {A\arabic{table}}
\numberwithin{equation}{section}
\numberwithin{lemma}{section}

\color{black}

\section{Extensions}


\subsection{Hodges--Lehmann point and interval estimation}
\label{section::hodgeslehmann}

We can use the proposed procedure to obtain a Hodges--Lehmann point estimate \citep{hodges1963estimates} and confidence intervals by inverting a sequence of tests. \citet{rosenbaum2002observational} gives a textbook discussion, and \citet{basse2019rand} describe both in the context of conditional randomization tests.
The Hodges--Lehmann point estimate of the difference in location between two vectors $x, y$ of dimension $n$ and $m$, respectively, is the median of all their pairwise differences: $(x_i - y_j)$, for all $i\in [n], j\in[m]$.

Extending the pairwise null hypotheses to allow for a non-zero constant effect
$$
H_0^{\h_1,\h_2}(c): \Yw_i(\h_1) - \Yw_i(\h_2) = c \text{ for all } i \in \mathbb{U},
$$
we can determine the potential outcome $\{ \Yw_i(w_1) , \Yw_i(w_2)\}$ for all focal units, i.e., units with $w_i(Z) = \h_1$ or $w_i(Z) = \h_2$. So conditional on $\Uobs$, if the test statistic depends only on the values of the outcomes with $w_i(Z) = \h_1$ or $w_i(Z) = \h_2$ as in Eq.~\eqref{eq:test-stat}, then  $H_0^{\h_1,\h_2}(c)$ acts as a sharp null hypothesis. We can therefore compute the corresponding $p$-value, denoted by $p(c)$, based on the conditional randomization test in Section \ref{section:rand-non-sharp}. By inverting 
our test, we can derive $\{ c: p(c) \geq \alpha   \}$ as the $(1-\alpha)$ level confidence interval for the constant treatment effect. In Appendix~\ref{appendix:HL}, we examine the coverage of these confidence intervals through realistic simulations.

\subsection{Relaxing Assumption~\ref{asst:properly-specified}}  
\label{section:assumptions-hypothesis}
To clarify the role of Assumption~\ref{asst:properly-specified}, we can restate our hypotheses using more general notation:
\begin{equation*}\label{eq:type-1-alt}
	\tilde{H}_0: Y_i(Z) = Y_i(Z') \text{ for all } Z,Z' \text{ and for all } i \in \Ib
\end{equation*}
and
\begin{equation*}\label{eq:type-1-alt}
	\tilde{H}_0^{\h_1,\h_2}: Y_i(Z) = Y_i(Z') \text{ for all } Z,Z' \text{ such that } w_i(Z),w_i(Z') \in \{\h_1,\h_2\} \text{ and for all }  i \in \Ib . 
\end{equation*}
If Assumption~\ref{asst:properly-specified} holds, the null hypotheses $\tilde{H}_0$ and 
$\tilde{H}_0^{\h_1,\h_2}$ are equivalent to the null hypotheses $H_0$ and 
$H_0^{\h_1,\h_2}$; if it does not hold, the null hypotheses $H_0$ and $H_0^{\h_1, \h_2}$ are 
not well defined, while $\tilde{H}_0$ and $\tilde{H}_0^{\h_1,\h_2}$ can still be tested. 
In fact, the procedures in Section~\ref{section:valid} used for testing $H_0$ and $H_0^{\h_1,\h_2}$ 
can be used without any modification to test $\tilde{H}_0$ and $\tilde{H}_0^{\h_1,\h_2}$ regardless of Assumption~\ref{asst:properly-specified}. 

While Assumption~\ref{asst:properly-specified} does not affect the mechanics of the test, it does impose restrictions on the alternative hypothesis, which changes the interpretation of rejecting the null hypothesis. 
In particular, Assumption~\ref{asst:properly-specified} imposes two levels of exclusion restriction: one on the relevant 
attribute and one on the relevant group. Without this assumption, a number of different reasons could lead to rejecting the null hypotheses, $H_0$ or 
$H_0^{\h_1,\h_2}$.
For instance, we would reject these hypotheses if a unit's outcome depends on the composition of attributes other than $A$, or if $A$ is the 
relevant attribute but a unit's outcome depends on the composition of groups other than its own. 
Assumption~\ref{asst:properly-specified} rules out both of these alternative channels for peer effects, narrowing the interpretation of rejecting the null hypotheses.

In summary, it is possible to test the null hypotheses $\tilde{H}_0$ and $\tilde{H}_0^{\h_1,\h_2}$ 
using the procedures in Section~\ref{section:valid}, regardless of the validity of 
Assumption~\ref{asst:properly-specified}. The price paid for the additional flexibility is that rejecting the null 
becomes less informative, since the alternative hypothesis includes channels of interference that were 
otherwise ruled out by Assumption~\ref{asst:properly-specified}.

As we discuss in the main text, there is little guidance for applied researchers on specifying exposure mappings, in part because these mappings can be highly context dependent. Thus, developing recommendations for exposure mappings in practice, as well as assessing sensitivity to those choices, is a necessary next step. 

\subsection{Testing weak null hypotheses}
\label{section::testweaknull}
Our paper focuses on null hypotheses that impose a constant effect (usually zero) for all units. A natural question is how to extend our approach to \emph{average} (or weak) null hypotheses. In the no-interference setting, \citet{wu2018randomization} propose permutation tests for weak null hypotheses using studentized test statistics. The result in \citet[][section 5.1]{wu2018randomization} suggests that our permutation tests in Section \ref{section:practical} can also preserve the asymptotic type I error under weak null hypotheses with appropriately chosen test statistics. For example, we can test the following weak null hypothesis
$$
H_{0}^{\h_1,\h_2}:  \tau(\h_1, \h_2) = 0 
$$
where $ \tau(\h_1, \h_2) =  N^{-1} \sum_{i=1}^N Y_i(\h_1) - N^{-1} \sum_{i=1}^N  Y_i(\h_2).$
Following the argument in \citet{wu2018randomization}, Procedure \ref{proc:nonsharp-permute} will deliver an asymptotically valid $p$-value for $H_{0  }^{\h_1,\h_2}$ if we use the studentized statistic 
$$
T(z;Y,\mathcal{U}) = \frac{   \sum_{a \in \Aset }  \pi_{[a] } ( \hat{\bar{Y}}_{[a]\h_1} -    \hat{\bar{Y}}_{[a]\h_2} ) }
{   \sqrt{  \sum_{a \in \Aset }  \pi^2_{[a] } (  \hat{S}^2_{[a]\h_1} /n_{[a]\h_1}   +  \hat{S}^2_{[a]\h_2} / n_{[a]\h_2}  )   }  },
$$
where   $\pi_{[a] }$ is the proportion of $A_i=a$ among all units $i\in \mathbb{U}$, and $( n , \hat{\bar{Y}}, \hat{S}^2)$ are the sample size, mean and variance with subscripts denoting the attribute and exposure. Coupled with the Hodges--Lehmann strategy in Section \ref{section::hodgeslehmann}, we can also construct asymptotic confidence interval for the average treatment effect $ \tau(\h_1, \h_2)$  by inverting permutation tests. Simulations in Section \ref{appendix:HL} confirm this empirically, and show that the 
resulting confidence intervals are indeed informative.

\subsection{Connection with the classic stratified, multi-arm trial}

Our paper helps to clarify the relationship between randomized group formation experiments and traditional randomized stratified experiments in settings without interference or peer effects. In particular, we show that the designs we consider are equivalent to classic stratified randomized experiments with multiple arms. The non-sharp null hypotheses of interest correspond to contrasts between different arms of a multi-arm trial, possibly for a subset of units. Thus, at least with some reasonable simplifying assumptions, the otherwise complex setting of randomized group formation experiments reduces to a more familiar setup. As a byproduct, our proposed permutation tests are applicable to the classic designs as well.

\section{Additional analysis for \citet{cai2017interfirm}}
\label{appendix:cai_additional}

This appendix section provides additional analysis and discussion of the re-analysis of \citet{cai2017interfirm} in Section \ref{section:cai}.

\subsection{Discussion of Assumption~\ref{asst:properly-specified} --- Alternative definitions of exposures}\label{sec6:assumption}

As discussed in Section~\ref{sec:assumptions}, the interpretation of our test hinges on $W$ being well-specified in the sense of Assumption~\ref{asst:properly-specified}.
For instance, our tests could reject, in principle, even if $H_0$ was true but firm revenues differed across group assignments that produced the same peer size exposure. Here, we explore the robustness of our results to 
two alternative specifications of the exposure. In the next section, we consider an additional specification, which reflects 
the type of peer group exposure that was actually randomized by \cite{cai2017interfirm}.

In particular, we consider two additional definitions of exposures:
$$
W_i^{(1)} = \frac{1}{|Z_i|}\sum_{j \in Z_i} \text{binary\_size}_j,~~\text{or}~~
W_i^{(2)} = \frac{1}{|Z_i|}\sum_{j \in Z_i} \text{size}_j \cdot \text{revenue}_j,
$$
where $\text{binary\_size}_j=1$ if and only if firm $j$ has size larger than the median size in $j$'s region; 
and $\text{revenue}_j$ is the log-revenue of firm $j$ at baseline. The definitions capture coarser or finer versions, respectively, of our original exposure. For both these definitions, we run Procedure~\ref{proc:sharp-W} and report the results in Table~\ref{tab:many_exposures} below.

From Table~\ref{tab:many_exposures}, we observe that our results remain largely robust to the alternative exposure specifications we consider. 
For instance, across all specifications, we find a significant effect on small service firms, as in the previous section. There is one notable difference, however. 
Under the coarser exposure definition,  $W_i^{(1)}$, we find evidence for a {\em negative} peer group effect on small manufacturing firms (two-sided $p$-value=0.04). This effect likely averages out the 
positive effect on small service firms~(two-sided $p$-value=0.011), and produces a nonsignificant overall effect under $W_i^{(1)}$.

\renewcommand{\arraystretch}{1.2}
\begin{table}[t!]
\centering
\caption{Testing the sharp null under alternative exposures. `one-sided' indicates the one-sided $p$-value ($p$) from Procedure~\ref{proc:sharp-W} on the relevant subpopulation; 
`two-sided' is the corresponding two-sided $p$-value, $2\min(p,1-p)$; `$*$' indicates a significant $p$-value at 5\% level.}
\label{tab:many_exposures}
\begin{tabular}{l rrrr}
\hline
                          & \multicolumn{2}{c}{$W_i^{(1)}$} &\multicolumn{2}{c}{$W_i^{(2)}$}           \\
                          & one-sided  & two-sided & one-sided  & two-sided \\
                          \hline
                          \hline
all firms                 & 0.83       & 0.34      & 0.037$^*$     & 0.075    \\
\hline
small service firms       & 0.0056$^*$     & 0.011$^*$     & 0.0018$^*$     & 0.0036$^*$    \\
small manufacturing firms & 0.98       & 0.04$^*$      & 0.54       & 0.92      \\
large service firms       & 0.61       & 0.78      & 0.25      & 0.51     \\
large manufacturing firms & 0.95       & 0.1       & 0.29      & 0.59    \\
\hline
\end{tabular}
\end{table}


\subsection{Pairwise null hypotheses}\label{sec:cs_pairwise}
We now turn to pairwise non-sharp null hypotheses, extending the analysis of heterogeneity in the previous section. 
To that end, we focus on  small manufacturing firms for which we observed a negative peer group effect in the previous section.
We also consider a definition of treatment exposure that matches the type of exposure  randomized in the actual experiment.

In particular, \cite{cai2017interfirm} randomized firms into 4 group types, namely, ``small firms in the same sector", ``large firms in the same sector", ``mixed-size firms in the same sector", and ``mixed-size firms with mixed sectors".  We thus 
define the following discrete-valued exposure for a small manufacturing firm $i$:
\begin{equation}\label{eq:w3}
W_i^{(3)} = \begin{cases}
  \tS,  & \text{if firm $i$'s peer group is all small manufacturing firms}; \\
    \tSm,  & \text{if firm $i$'s peer group is all small firms of various sectors}; \\
  \tSL, & \text{if firm $i$'s peer group is mixed-size manufacturing firms}; \\
  \tSLm, & \text{if firm $i$'s peer group is mixed-size firms of various sectors.}
\end{cases}
\end{equation}

We consider four (weak) pairwise null hypotheses each comparing whether small manufacturing firms benefit from having a certain exposure level over another. 
For instance,  $H_0^{\tS, \tSL}(\text{small})$ denotes a null hypothesis 
to test whether there 
are benefits of having a mix of large and small manufacturing peers as opposed to having only small manufacturing peers; 
 $H_0^{\tS, \tSm}(\text{small})$ denotes whether there 
are benefits of having a mix of small service or small manufacturing peers as opposed to having only small manufacturing peers; and so on.
%

Table~\ref{table:cai-results} summarizes the results from using Procedure~\ref{proc:general-W} on these pairwise null hypotheses.
These results adds nuance to the negative peer group effect that we observed on small manufacturing firms in Table~\ref{tab:many_exposures}.
In particular, we find that this negative peer group effect on small manufacturing firms is mainly due to their exposure to other large manufacturing firms. The relevant null, $H_0^{\tS, \tSL}$, is strongly rejected (two-sided $p$-value= 0.008), and the inverted confidence interval from this test indicates a range of 15\% to 65\% in revenue loss from such exposure.
In contrast, no negative effects are observed when the exposure of small manufacturing firms is to small or large firms from a different sector (service).

\renewcommand{\arraystretch}{1.5}
\begin{table}[t!]
\centering
\caption{Two-sided $p$-values and inverted randomization-based confidence intervals (at 5\% level) for the pairwise weak nulls of Section~\ref{sec:cs_pairwise}. `$n$' indicates the number of units tested under the respective null, $H_0^{\h_1, \h_2}$; 
`$n_1$' is the number of firms exposed to $\h_1$, and `$n_2$ the number 
of firms exposed to $\h_2$~($n=n_1+n_2$).} 
	\label{table:cai-results}
	\begin{tabular}{|c|c|c|c|c|}
		\hline
		Null hypothesis & $n~(n_2/n_1)$ & $p$-value & point estimate & confidence interval \\
		\hline
		$H_0^{\tS, \tSL}(\text{small})$ & 179 (84/95) & 0.008 & -0.45 & (-1.05, -0.17) \\
  $H_0^{\tS, \tSm}(\text{small})$ & 139~(44/95) &  0.76 & -0.55 & (-1.10, 0.84) \\
    $H_0^{\tSL, \tSLm}(\text{small})$ & 188 (104/84) &  0.91 & 0.016 & (-0.44, 0.41) \\
  $H_0^{\tSm, \tSLm}(\text{small})$ & 148 (104/44) &  0.29 & 0.11 & (-1.24, 0.40) \\
		\hline		
	\end{tabular}
\end{table}

\section{Simulation studies}
\label{appendix:simulations}

This appendix section includes three sets of simulation studies. The first two demonstrate the failure of asymptotic approximations in our applications and highlight the importance of using exact tests. The last set of simulations explores the power of the proposed randomization test.

\subsection{Simulation study calibrated to~\citet{li2018randomization}}
\label{appendix:li}
Our first simulation study illustrates the failure of asymptotics of the regression-based (``Neymanian'') approach proposed by \citet{li2018randomization}, in a setting calibrated to the roommates application in Section \ref{section:li}.
Specifically, consider the following setup:
\begin{itemize}
\item  $N=156$ students allocated at random in rooms of size 4, indexed by $i$.
\item  A random $a\%$ of students ($a$ is a free parameter) has $A=1$ and the rest has $A=0$. 
\item Sample $X_i\sim N(0,1)$ iid; or $X_i = S_i\mathrm{Weibull(0.3)}$, where $S_i$ is random sign; 
or $X_i \sim \mathrm{mixture}$ where $\mathrm{mixture} = (1-B) \delta_{-k} + B U[1-\epsilon, 1+\epsilon]$, where 
$\delta$ is the delta function, $k, \epsilon$ are constants and $B$ is a Bernoulli random variable such that the mean is 0.
All distributions are also normalized to have variance 1.

\item Sample $\varepsilon_i$ iid using the distributions described above.
\item Define the exposure model, $W_i = \sum_{j\in \text{room}_i, j\neq i} A_{j}$, where $\text{room}_i$ is the set of students in the same room as $i$.
\item Define outcome $Y_i = 1 + 0\cdot \iv(W_i=2) + X_i+ (0.01+A_i) \varepsilon_i$.
\end{itemize}

Note that, under this DGP, $\Yw(0) = \Yw(2)$ in distribution, and so our randomization tests remain finite-sample valid.

In this model, even though room allocation is completely randomized and there is no imbalance in room size, the joint distribution of $(A,W)$ 
has a complex correlation structure due to the group formation design. 
In particular, roughly 3-5\% of the units are exposed to $W=2$, which results in a highly leveraged exposure assignment.
 Moreover, conditional on $W_i=2$, unit $i$ is more likely to be $A_i=0$. Thus, under a mixture error distribution the outcomes $Y_i$ of such units tend to be smaller than the outcomes under other exposures. This difference becomes negligible in the limit with more samples, but it is substantial in finite-samples, and cannot be easily captured by a regression model even under a robust specification.

To illustrate this point, we regress $Y_i \sim 1(W_i=2) + X_i$ and use conservative heteroskedasticity-robust errors (``HC0"). We then 
test (at 5\% level) the hypothesis that the regression coefficient of the exposure dummy variable is zero. A partial set of our results is shown in the table below.
 Here, we want only to show the  pathological cases for the regression approach, and so we exclude the normal error setting for which regression  performs well and near the nominal level.

\renewcommand{\arraystretch}{1}
\begin{table}[tbp!]
\centering
\caption{Rejection rates from robust regression based on a simulation motivated by Li et al. (2019).}
\label{table:appendix_li}
\begin{tabular}{rr cc}
  \hline
  $a$ (\%$A=1$) & $X$  & $\varepsilon$ & Rejection rate\% \\ 
  \hline
10.00 & $N(0,1)$ & $\mathrm{Weibull}$ & 1.11  \\ 
  30.00 &  &  & 1.60  \\ 
  50.00 &  &  & 1.60  \\ 
\hline
  10.00 & $\mathrm{Weibull}$ &  $\mathrm{Weibull}$ & 2.07  \\ 
  30.00 & &  & 1.40 \\ 
  50.00 &  & & 2.40  \\ 
\hline
  10.00 & $\mathrm{mixture}$ &  $\mathrm{Weibull}$  & 1.87  \\ 
  30.00 & &  & 1.50  \\ 
  50.00 &  & & 1.50  \\ 
\hline
  10.00 & $N(0,1)$ & $\mathrm{mixture}$ & 60.79  \\ 
  30.00 & &  & 11.49  \\ 
  50.00 &  &  & 10.40  \\ 
\hline
  10.00 & $\mathrm{Weibull}$ &  $\mathrm{mixture}$ & 63.84  \\ 
  30.00 &  &  & 10.90  \\ 
  50.00 &  &  & 10.50  \\ 
\hline
  10.00 & $\mathrm{mixture}$ &  $\mathrm{mixture}$ & 66.67 \\ 
  30.00 &  &  & 8.40  \\ 
  50.00 &  &  & 10.40 \\ 
   \hline
   \hline
\end{tabular}
\end{table}

Based on the results reported in Table~\ref{table:appendix_li}, we observe that 
with Weibull errors (heavy tailed), the regression-based test has a size distortion and tends to under-reject. Under a mixture distribution for the errors, regression severely over-rejects. For instance, even with $N(0,1)$ covariates, we observe rejection rates up to roughly 61\%.
In general, the regression-based test deteriorates under imbalanced designs.

In contrast, the randomization test is finite-sample valid as expected. Table \ref{table:appendix_li} shows a partial set of results relating to the pathological cases. We see that the randomization test achieves near-nominal level performance, with deviations from the nominal level due to Monte Carlo error.

\begin{table}[tbp!]
\centering
\caption{Rejection rates from robust regression and the group formation randomization test of Procedure~\ref{proc:general-W}.}
\label{table:appendix_li}
\begin{tabular}{rr ccc}
  \hline
  $a$ (\%$A=1$) & $X$  & $\varepsilon$ & regression & randomization test \\ 
  \hline
  10.00 & $N(0,1)$ & $\mathrm{mixture}$ & 60.79 & 4.96  \\ 
  30.00 & &  & 11.49 & 5.53 \\ 
  50.00 &  &  & 10.40  &  5.13 \\ 
\hline\hline
\end{tabular}
\end{table}
\subsection{Simulation study calibrated to \citet{cai2017interfirm}}
\label{appendix:cai_sim}

We now consider the following simulation setup inspired by the analysis of \citet{cai2017interfirm} in Section \ref{section:cai}. Here we focus on a subset of the data to illustrate the key intuition.
We have 13 firms in the same sector and subregion, 2 of the firms are ``large'' and the remainder are ``small.'' In particular, their sizes in terms of log number of employees are
$A = (5, 5, Z_1, \ldots Z_{11})$ where $Z_i\sim \mathrm{Unif}[1,3]$ are iid uniform. Following~\cite{cai2017interfirm} we randomize the firms into two 
groups, one of type ``mixed-size" (SL) and another of type ``small-size" (S). 
Since $Z_i$ are iid we can simply set  as $L = (1, 1, 1, 2, 2, \ldots, 2)$, such that 
group 1 is of type (SL) with two large firms and one small firm, and group 2 is of type (S) with all firms being small.
The exposure of firm $i$ is defined as the average group size of other firms in $i$'s group:
$$
W_i = \frac{1}{|\text{group}_i|} \sum_{j\in\text{group}_i} A_j.
$$
We sample $\epsilon_i = N(0,\sigma_i^2)$ where $\sigma_i^2 = 1/|\text{group}_i|$ is the reciprocal of $i$'s group size, and set the outcome model
as $Y_i = 0\cdot W_i + \epsilon_i$. 

A conventional econometric approach, such as a regression inspired by a linear-in-means model, is to regress $Y \sim W$ or $Y \sim W + A$ and test whether the coefficient on $W$ is zero. 
However, both approaches are severely biased even when we condition on the same sector, subregion and firm sizes.
In a simulated study with 10,000 replications based on this model, the nominal 5\% rejection rate from the regression $Y \sim W$ is 15.5\%; and the rejection rate from the regression $Y \sim W + A$ is 17.6\%.

%
The problem here is that uncertainty quantification under linear regression does not 
take into account the correlation structure in $W$. For instance, in this model, both large firms have the exact same exposure regardless of the particular treatment assignment. Due 
to the problem structure, with high probability the errors in these two large groups can both be extreme leading to a spurious correlation between $Y$ and $W$. 
Conditioning on firm characteristics in a regression model cannot fix this issue. In contrast, a randomization test can leverage the true correlation structure in $W$ and has the correct level in finite-samples.

\subsection{Power simulations}\label{appendix:power}
Here, we provide simulation results on the power
of our tests. Our simulation setting mimics our roommate application in Section \ref{section:li}. We consider $N = 156$ units split into
groups of four units. As in the application, $104$ of these units have attribute
$A=1$, while the rest have attribute $A=0$. The peer group exposure for unit $i$ is the sum 
of group peers with $A=1$ --- thus, $\Wdom = \{0, 1, 2, 3\}$.
Throughout the simulation, we will
focus on testing the null hypothesis $H_0^{1,0}: \Yw_i(1) = \Yw_i(0)$, so our simulation
setup only needs to specify the potential outcomes $\Yw_i(0)$ and $\Yw_i(1)$.  

\subsubsection{No covariates}
\label{section:simul-nocov}

We first simulate IID potential outcomes from
\begin{flalign*}
  \Yw_i(0) &= 4 \times \text{Beta}(10, 3),  \\
  \Yw_i(1) &= \min\{ \Yw_i(0) + \tau, 4 \} . 
\end{flalign*}
With these specifications, the potential outcomes live on a $4$ points scale,
like the original GPA outcomes, and the mean is in the same ballpark as the
original data. When $\tau = 0$, the null hypothesis is true so we expect a
rejection rate at the nominal level. When $\tau \neq 0$, the alternative is
true: to study the power, we generated potential outcome with
$\tau$ taking the values $\tau = 0, 0.1, 02, 0.3, 0.4, 0.5, 0.6, 0.7, 0.8, 1$.
For each value of $\tau$ we generate a schedule of potential outcomes. We then
generate 300 draws of $Z^\text{obs}$ using an stratified randomized design. 
Then for each $Z^\text{obs}$, we run our test to obtain
a $p$-value with 1000 draws
from the condition distribution using the difference-in-means statistic. Finally, we compute the rejection rate over all draws at level   $\alpha = 0.05$. Figure~\ref{fig::simulation1} summarizes
the results.
\begin{figure}[h]
  \centering
  \includegraphics[scale=0.35]{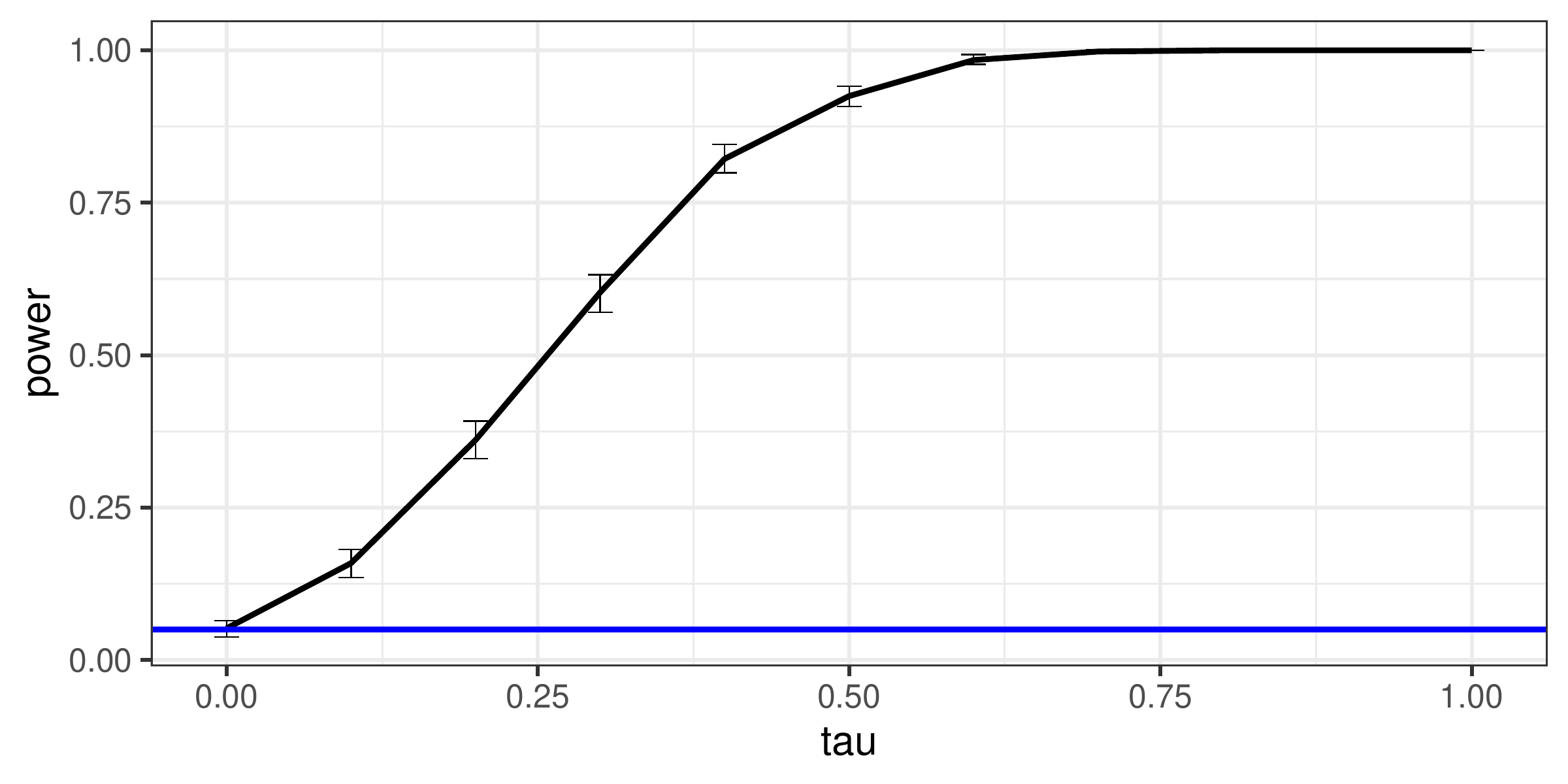}
  \caption{Power of test without using covariates}\label{fig::simulation1}
\end{figure}
As expected,   the rejection rate at $\tau = 0$ is equal to
$0.05$, while the power
increases as $\tau$ increases. For $\tau > 0.7$, the power is essentially
1. The power reaches $0.5$ at $\tau = 0.25$.

\subsubsection{Leveraging covariate information}

In our second set of simulations, we illustrate the power gains that can be
obtained by stratifying on both the attribute of interest and additional
covariates. In this section, in addition to an attribute $A_i$, each unit $i$
has a covariate binary $X_i$. We simulated data so that half of the unit with
attribute level $A_i=1$ has covariate value $X_i=0$ and half has the value $X_i = 1$,
and similarly for the units with attribute level $A_i=0$. That is, 
$\sum_{i=1}^N A_i X_i = \sum_{i=1}^N A_i (1-X_i) = 52$ and
$\sum_{i=1}^N (1-A_i) X_i = \sum_{i=1}^N (1-A_i) (1-X_i) = 26$. We then simulate IID  potential outcomes from 
\begin{align}
  \Yw_i(0) &= 4 \times \{ (1-X_i) \text{Beta}(10, 3) + X_i \text{Beta}(5, 5)\} \nonumber \\
  \Yw_i(1) &= \min\{ \Yw_i(0) + \tau, 4 \}  \label{app:y1}
\end{align}
That is, the distribution of the control potential outcomes is different for
the two values of the covariate. 
For each value of $\tau$ we generate a schedule of potential outcomes, then
we take two approaches:
\begin{enumerate}
\item We generate 300 draws of $Z^\text{obs}$ using a stratified randomized design that stratifies
  on both the attribute and $X_i$. Then for each $Z^\text{obs}$, we run our
  test to obtain a $p$-value. Finally, we compute the rejection rate over all  draws at level $\alpha = 0.05$. 
\item We do the same as above, but stratifying only on the attribute, not on $X_i$.
\end{enumerate}

\begin{figure}[h]
  \includegraphics[scale=0.35]{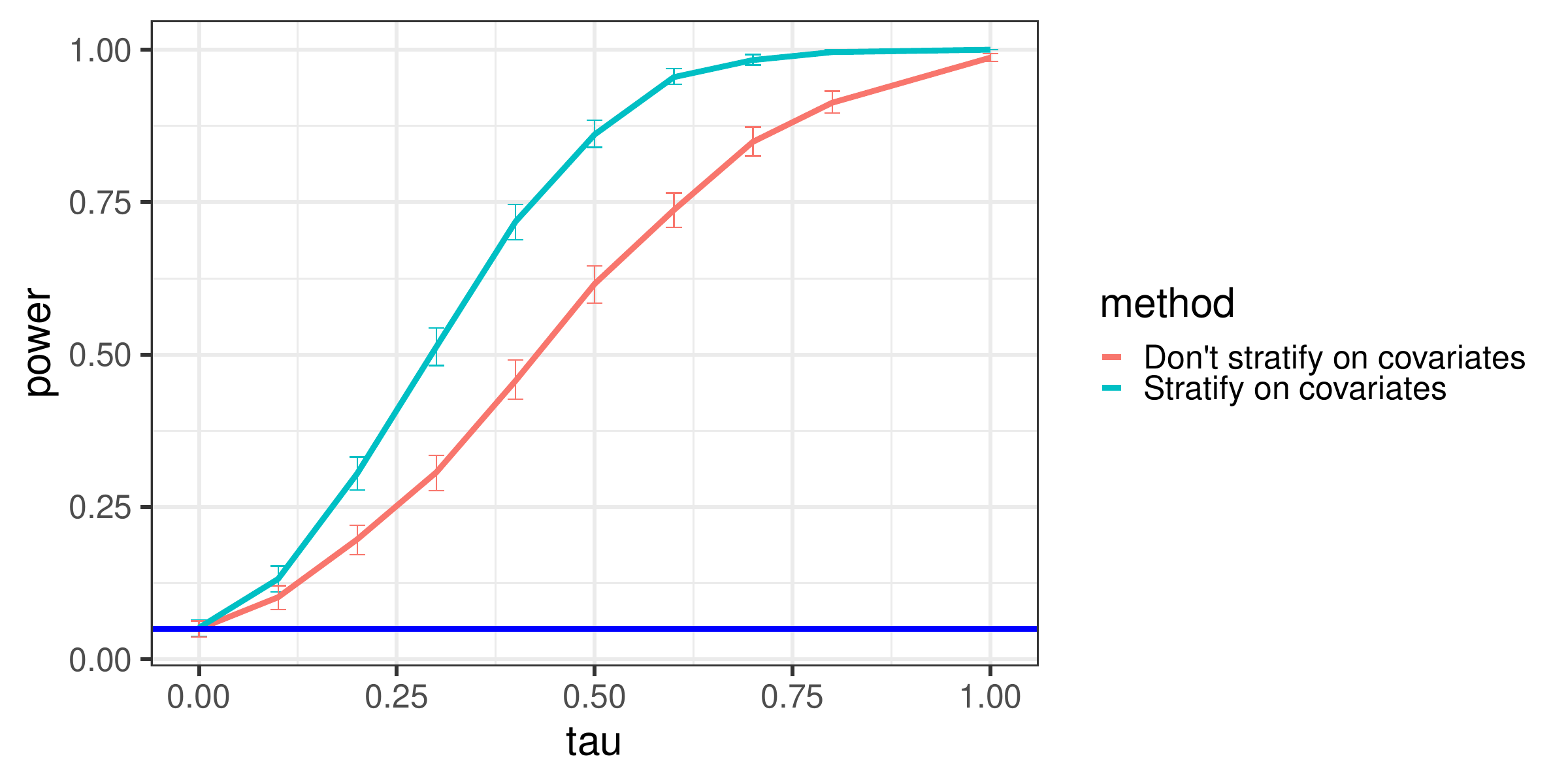}
  \caption{Studying the power gains from stratifying on covariates}
  \label{fig:cov}
\end{figure}

Figure~\ref{fig:cov} summarizes the results. 
Both methods the rejection rate is equal to the
nominal rate when $\tau=0$, and the power increases with $\tau$. The power is much higher when leveraging covariate information. For instance, when
$\tau \approx 0.35$, the power of the test leveraging covariates is about $0.7$,
while that of the test ignoring covariates is below $0.5$.

\subsubsection{Examining the Hodges--Lehmann estimator}\label{appendix:HL}
We also ran simulations with the same setup as in Section~\ref{section:simul-nocov}, but 
this time we examined the properties of the Hodges--Lehmann estimator, computed using the 
studentized test statistic. This setting is of particular interest because the effect is 
not constant and additive. For a given value of $\tau$, the average treatment effect is:
\begin{equation*}
    \tau^\ast = \frac{1}{N} \sum_{i=1}^N \{ \Yw_{i}(1; \tau) - \Yw_i(0)\},
\end{equation*}
where $\Yw_{i}(1; \tau)$ is the `treated' potential outcome under Definition~\eqref{app:y1}.
The discussion of Section~\ref{section::testweaknull} predicts 
that the confidence interval is valid asymptotically. Table~\ref{tab:hl-simuls} below reports 
the coverage, the size of the interval, and the fraction of intervals that do not contain 
$0$ (akin to the ``power''). The results confirm our theory. We see that the coverage is nominal, 
and that our the size of our confidence intervals is reasonable ($80\%$ of intervals for 
$\tau=0.5$ fail to cover $0$).

\begin{table}[]
    \centering
        \caption{Summary of simulations for the Hodges--Lehmann estimator.}
    \label{tab:hl-simuls}
    \begin{tabular}{c|c|c|c|c|c|c|}
        $\tau$ & 0 & 0.1 & 0.2 & 0.3 & 0.4 & 0.5  \\
        \hline
        $\tau^\ast$ & 0 & 0.1 & 0.198 & 0.291 & 0.378 & 0.458 \\
        coverage & 0.96 & 0.96 & 0.94 & 0.94 & 0.97 & 0.92 \\
        Interval length & 0.65 & 0.63 & 0.62 & 0.61 & 0.59 & 0.56 \\
        Power & 0.035 & 0.06 & 0.14 & 0.34 & 0.54 & 0.80 \\
        \hline
    \end{tabular}
\end{table}

\clearpage
\color{black}

\newcommand{\nns}{\mathbf{m}^s}
\newcommand{\nnc}{\mathbf{m}^c}

\newcommand{\BB}{\mathcal{B}}
\newcommand{\commenteq}[1]{\text{\footnotesize #1}}
\section{Proofs}\label{sec::proofs}
\subsection{Proof of Theorem~\ref{thm1}}
\begin{theorem}
    \TheoremOne
\end{theorem}
\begin{proof}
We start with two lemmas.
\begin{lemma} \label{lemma:proof1}
Suppose that Conditions (a)--(c) of Theorem~\ref{thm1} hold. 
Let $\BB\in \mathcal{O}(\Wdom^N; \SAU)$ be an orbit such that $P(\BB) > 0$.
Then, for any $\pi\in\SAU$, we have 
$$
P(\pi L  \mid \{W\in\BB\}, U) = P(L \mid \{W\in\BB\}, U).
$$
\end{lemma}
\begin{proof}[Proof of Lemma \ref{lemma:proof1}]
$L$ determines both $U$ and $W$, and so
\begin{equation}\label{proof1_1}
P(W \in \BB, U \mid L) = \iv\{ \wl(L)\in \BB\} \cdot \iv\{U = \ul(L)\}.
\end{equation}
Similarly,
\begin{align}\label{proof1_2}
P(W \in \BB, U \mid \pi L) & =  \iv\{ \wl(\pi L)\in \BB\} \cdot \iv\{U = \ul(\pi L)\} 
\quad\quad\commenteq{from~\eqref{proof1_1}}
\nonumber\\
& = \iv\{ \pi \wl( L)\in \BB\} \cdot \iv\{U = \pi \ul(L)\} 
\quad\quad\commenteq{from Conditions (b)-(c)}\nonumber\\
& =  \iv\{ \wl( L)\in \BB\} \cdot \iv\{\pi^{-1} U = \ul(L)\} 
\quad\quad\commenteq{from orbit property of $\BB$} \nonumber\\
& =  \iv\{ \wl( L)\in \BB\} \cdot \iv\{ U = \ul(L)\} 
\quad\quad\commenteq{$\pi U = U$ since $\pi\in\SAU$ } \nonumber\\
& = P(W \in \BB, U \mid L) . 
\quad\quad\commenteq{from~\eqref{proof1_1}}
\end{align}
It follows that
\begin{align}
&P(W\in\BB, U \mid \pi L) P(\pi  L)  = P(W\in\BB, U \mid L) P( L) , \quad\quad\commenteq{From~\eqref{proof1_2} and Condition (a)} \nonumber\\
\Rightarrow &
\frac{P(W\in\BB, U \mid \pi L) P(\pi L)}{P(\BB)}  = \frac{P(W\in\BB,U \mid L) P(L)}{P(\BB)} , \quad\quad\commenteq{From $P(\BB) > 0$ } \nonumber\\
\Rightarrow &
P(\pi L \mid \{W\in\BB\}, U)  = P( L \mid \{W\in\BB\}, U).\nonumber
\end{align}
\end{proof}
Lemma~\ref{lemma:proof1} shows that $L$ retains its symmetry even conditionally on $W$ beloning to some orbit $\mathcal{B}$ and conditional on focal selection $U$. The subspace where its symmetry holds is exactly the permutation subgroup $\SAU$, which leaves $A$ and $U$ fixed.

\begin{lemma}\label{lemma:proof1_2}
Let $\h\in\Wdom^N$ be a fixed exposure vector, and define $$
\Ldom(\h) = \{ L\in\Ldom: \wl(L) = \h\}.
$$
Then,  for any $\pi\in\SAU$, we have that
$$\Ldom(\pi \h) = \{ \pi L : L\in \Ldom(\h)\}.$$
\end{lemma}
\begin{proof}[Proof of Lemma \ref{lemma:proof1_2}]
The result follows from the equivariance property of $\wl$ in Condition (b).
Specifically, equivariance implies that for any $L \in \Ldom(\h)$ then $\pi L\in \Ldom(\pi \h)$. 
Conversely, for any $L'\in \Ldom(\pi \h)$ then $\pi^{-1} L' \in \Ldom(\h)$. 
\end{proof}
The crucial result in Lemma~\ref{lemma:proof1_2} is that there exists a 1-1 mapping between the sets $\Ldom(\h)$ and $\Ldom(\pi\h)$ for any $\pi\in\SAU$.

We are now ready to prove the main result of Theorem~\ref{thm1}. For a fixed $\h\in\Wdom^N$:
\begin{align}\label{proof1_3}
P(W = \h \mid \BB, U ) =  \sum_{L\in\Ldom} \iv\big\{\wl(L) = \h\big\} P(L \mid \BB, U) = \sum_{L\in\Ldom(\h)} P(L \mid \BB, U),
\end{align}
where ``$|\BB, U$" is shorthand for conditioning on event ``$\{W \in \BB\}, U$". 
Moreover, for any $\pi \in\SAU$:
\begin{align}\label{proof1_4}
P(W = \pi\h \mid \BB, U)  & =   \sum_{L\in\Ldom} \iv\big\{\wl(L) = \pi\h\big\} P(L \mid \BB, U) 
\quad\quad\commenteq{From~\eqref{proof1_3}}\nonumber\\
& =\sum_{L\in\Ldom(\pi \h)} P(L \mid \BB, U)  \nonumber\\
& = \sum_{L\in \Ldom(\h)} P(\pi L \mid \BB, U) 
\quad\quad\commenteq{From Lemma~\ref{lemma:proof1_2}}
\nonumber\\
& = \sum_{L\in \Ldom(\h)} P(L \mid \BB, U) 
\quad\quad\commenteq{From Lemma~\ref{lemma:proof1}}
\nonumber\\
&=  P(W=\h \mid \BB, U).
\end{align}
$\BB$ is an orbit, and so it can be generated by any of its elements. 
Since $W\in\BB$, the orbit can be generated by $W$, and so $\BB = \{\pi W: \pi\in\SAU\}$.
Therefore, conditional on $\{W\in\BB\}$ and focals $U$, the orbit $\BB$ is the entire domain of $W$. The result in~\eqref{proof1_4} now implies that 
$W$ is conditionally uniform given $\BB$ and $U$.
\end{proof}

\subsection{Proof of Lemma \ref{lemma::conditions-b-c}}\label{proof:conditions}

{\bf Equivariance of $\wl$.}~The exposure is defined in Eq.~\eqref{eq:exposure} as $w_i(Z) = \{ A_j : j \in Z_i\}$. On the domain of group levels, this can be re-written as:
$$
\wl_i(L) = \{A_j : L_j=L_i, j\neq i\}.
$$
Now, let $\pi\in\SN(A)$ be any transposition acting on $L$, i.e., a single swap between labels 
$L_i, L_j$ of units $i$ and $j$, respectively. After the swap, $i$ is in the ``room" that $j$ was, and $j$ is in the ``room" that $i$ was. From the 
definition of $\wl$ above, the exposures are only a function of other units' attributes in the room, and so units $i$ and $j$ swap exposures.
The exposures of all units other than $i,j$ are unaffected because $i$ and $j$ have the same attribute ($A_i = A_j$) due to $\pi\in\SN(A)$.

Thus, we proved that $\wl(\pi L) = \pi \wl(L)$ whenever $\pi$ is a transposition. Since every permutation is a composition of transpositions,  the result holds for any permutation in $\SN(A)$. Moreover, the result holds for $\pi\in\SAU$ as well since $\SAU$ is a subgroup of $\SN(A)$. 

{\bf Equivariance of $\ul$.}~Recall the definition of focal selection in our setting, as defined in Eq.~\eqref{eq:focals-1}, 
$u_i(Z) = 1$ if and only if $w_i(Z) \in \{\h_1, \h_2\}$. 
With a slight abuse of notation, this can be re-written as $\ul(L) = \iv\{ \wl(L) \in\{\h_1,\h_2\} \}$, where the operation on the  right-hand side is understood element-wise.
Thus, 
$$
\ul(\pi L) =\iv\{ \wl(\pi L) \in\{\h_1,\h_2\} \} = 
 \iv\{\pi \wl( L) \in\{\h_1,\h_2\} \} =  \pi \iv\{ \wl( L) \in\{\h_1,\h_2\} \}.
$$
Here, the second equality follows from equivariance of $\wl$ and the last equality follows from the element-wise operation.

\subsection{Proofs of Condition~(a) for the designs in Section~\ref{sec:perm_designs}}\label{proof:stratified}


{\bf Proof for stratified randomized design.}
In the stratified randomized design,  define $\nns: \Lset^N \to \mathbb{N}^{|\Aset| \times |\Lset|}$ as
$$
\nns(L)_{a,k} = \sum_{i\in\Ib} \iv(L_i=k) \iv(A_i=a) , 
$$ 
which counts how many units with attribute $A_i=a$ are assigned to group label $k$. Then, a stratified randomized satisfies $\pr (L) \propto \iv\{\nns(L) = \nak_A\}$, where $\nak_A$ is fixed.
For any permutation $\pi\in\SN(A)$, and any pair $(a,k)$, we have 
\begin{align}
\nns(\pi L)_{a,k} & =  \sum_{i\in\Ib} \iv\{ (\pi L)_i=k\} \iv(A_i=a)  \nonumber\\
& = \sum_{i\in\Ib} \iv(L_i=k) \iv\{(\pi A)_i=a\} 
\quad\quad\commenteq{From identity, $(\pi x)'y = x'(\pi y)$, for any $x,y\in\mathbb{R}^N$ }  \nonumber\\ 
& = \sum_{i\in\Ib} \iv(L_i=k) \iv(A_i=a) 
\quad\quad\commenteq{ $\pi A = A$ since $\pi\in\SN(A)$ }  \nonumber\\ 
& = \nns(L)_{a,k}.
\end{align}
This results immediately implies that $P(\pi L) = P(L)$ for any $\pi\in\SN(A)$. This holds also in the focal selection setting. That is, $P(\pi L) = P(L)$ for any $\pi\in\SAU$ since 
$\SAU$ is a subgroup of $\SN(A)$. Thus, Condition (a) holds.

{\bf Proof for the completely randomized design.}
In the completely randomized design,  define $\nnc: \Lset^N \to \mathbb{N}^{|\Lset|}$ 
as 
$$
\nnc(L)_{k} = \sum_{i\in\Ib} \iv(L_i=k) , 
$$
which counts how many units  are assigned to group label $k$. Then, $\pr (L) \propto \iv\{\nnc(L) = \nak\}$, where $\nak = (n_1, \ldots, n_K)$ denotes how many units are to be assigned to each label, and is fixed.
For any permutation $\pi\in\SN$ and label $k$, we have 
\begin{align}
\nnc(\pi L)_{k} & =  \sum_{i\in\Ib} \iv\{ (\pi L)_i=k\} 
=  \sum_{i\in\Ib} \iv\{ L_i=k\} = \nn(L)_k.
\end{align}
This results immediately implies that $P(\pi L) = P(L)$ for any $\pi\in\SN$. This holds also for any subgroup of $\SN$, including $\SN(A)$ and $\SAU$. Both of these subgroups 
keep the attributes fixed, and so Procedures~\ref{proc:sharp-permute} and~\ref{proc:nonsharp-permute} in the completely randomized design are equivalent to the stratified randomized design 
with parameter $\nak_A = \nns(L)$.  Thus, Condition (a) holds. 

\end{document}